\documentclass[oneside,english]{amsart}
\usepackage[T1]{fontenc}
\usepackage[latin9]{inputenc}
\usepackage{geometry}
\usepackage{amsthm}
\usepackage{amstext}
\usepackage{amssymb}
\usepackage{graphicx}
\usepackage{setspace}
\usepackage{enumerate}

\usepackage{epsfig}
\usepackage{psfrag}
\usepackage{color, epsf}
\usepackage{tikz} 
\usepackage{caption}
\usepackage{subcaption}

\newcommand{\E}{\mathbb E}
\newcommand{\R}{\mathbb R}

\newcommand{\sgn}{\mbox{sgn}}

\numberwithin{equation}{section}
\numberwithin{figure}{section}

\theoremstyle{plain}
\newtheorem{thm}{\protect\theoremname}
  \theoremstyle{remark}
  \newtheorem*{rem*}{\protect\remarkname}
  \theoremstyle{plain}
  \newtheorem{lem}[thm]{\protect\lemmaname}
  \theoremstyle{plain}
  \newtheorem{prop}[thm]{\protect\propositionname}
  \theoremstyle{plain}
  \newtheorem{cor}[thm]{\protect\corollaryname}
   \theoremstyle{plain}
  \newtheorem*{thm*}{\protect\theoremname}
  
  \theoremstyle{plain}
    \theoremstyle{remark}

\theoremstyle{plain}

\makeatother

  \providecommand{\corollaryname}{Corollary}

  \providecommand{\lemmaname}{Lemma}
    \providecommand{\propositionname}{Proposition}
  \providecommand{\remarkname}{Remark}
  \providecommand{\theoremname}{Theorem}
\providecommand{\theoremname}{Theorem}
\providecommand{\resultname}{Result}
\providecommand{\definitionname}{Definition}

\title{Optimal consumption and investment under transaction costs}
\author{David Hobson \and Alex Tse \and Yeqi Zhu}
\thanks{David Hobson, Alex Tse: Department of Statistics, University of Warwick, Coventry, CV4 7AL, UK. D.Hobson@warwick.ac.uk; \\
Alex Tse: Cambridge Endowment for Research in Finance, Judge Business School, University of Cambridge, Cambridge, CB2 1AG, UK. S.Tse@jbs.cam.ac.uk; \\
Yeqi Zhu: Credit Suisse, London, UK. (The opinions expressed in the paper are those of the author and not of Credit Suisse.) Yeqi.Zhu@credit-suisse.com}
\date{\today}
\usepackage[english]{babel}

\begin{document}
\begin{abstract}

In this article we consider the Merton problem in a market with a single risky asset and transaction costs. We give a complete solution of the problem up to the solution of a free-boundary problem for a first-order differential equation, and find that the form of the solution (whether the problem is well-posed, whether the problem is well-posed only for large transaction costs, whether the no-transaction wedge lies in the first, second or fourth quadrants) depends only on a quadratic whose co-efficients are functions of the parameters of the problem, and then only through the value and slope of this quadratic at zero, one and the turning point.

We find that for some parameter values and for large transaction costs the location of the boundary at which sales of the risky asset occur is independent of the transaction cost on purchases. We give both a mathematical and financial reason for this phenomena.

\end{abstract}

\maketitle

\section{Introduction}
\label{sec:intro}

In this article we consider the problem of maximising expected utility of consumption over the infinite horizon in a financial market consisting of a riskless bond and a single risky asset. Merton~\cite{Merton:69,Merton:71} considered this problem in a perfect, frictionless market and showed that the optimal strategy is to keep a constant fraction of wealth in the risky asset. Of the possible market frictions, arguably the most significant is transaction costs. This paper adds to the growing literature on optimal consumption/investment problems with proportional transaction costs.

In the Merton setting the market is complete. Constantinides and Magill~\cite{ConstantinidesMagill:76} generalised the problem to the incomplete case by introducing transaction costs. They argued that in the case of a single risky asset following exponential Brownian motion and power utility the scalings of the problem should mean that there is a no-transaction wedge, and the optimal strategy should be to trade in a minimal fashion so as to keep the fraction of wealth in the risky asset within an interval $I$. If the initial portfolio is such that the initial fraction of wealth in the risky asset is outside this interval then the agent makes an instantaneous transaction to bring the fraction of wealth in the risky asset to the closest boundary of $I$. Thereafter, the agent only trades when this fraction is on the boundary of $I$. In $(\mbox{Cash wealth}, \mbox{Wealth in risky asset})$ space, the interval $I$ becomes a no-transaction wedge.

The model and accompanying intuition was formulated precisely by Davis and Norman~\cite{DavisNorman:90}. They used the language of stochastic control and martingales
to give a rigorous description of the problem for power utility. For a certain subset of parameter combinations they gave a full description of the solution, specifying both the optimal consumption, and the optimal investment strategy. The optimal investment strategy involves a process which receives a local-time push at both boundaries of an interval, and these pushes are just sufficient to keep the process within the interval. Davis and Norman~\cite{DavisNorman:90} reduce the problem to solving a pair of first-order ordinary differential equations (ODEs) subject to value matching conditions at unknown free-boundaries. Their analysis was extended by Shreve and Soner~\cite{ShreveSoner:94} to a larger class of parameter combinations using methodologies from viscosity solutions.

Both \cite{DavisNorman:90} and \cite{ShreveSoner:94} consider the value function and the primal problem. Recently, there have been a trio of papers considering the dual problem and shadow prices. Kallsen and Muhle-Karbe~\cite{KallsenMuhle-Karbe:10} were the first to use the shadow-price approach in this context, but only consider logarithmic utility. Herczegh and Prokaj~\cite{HerczeghProkaj:15} extend their results to power utility. The most complete treatment of the problem via the shadow-price approach is the paper of Choi et al~\cite{ChoiSirbuZitkovic:13}. Choi et al give a full analysis of the problem, covering all parameter combinations (which involve an appreciating risky asset). They reduce the problem to the solution of a free-boundary problem for a single first-order ODE. There are multiple solutions to this free-boundary problem, and the one that is wanted is the one for which the solution to the free-boundary problem satisfies an integral condition.

In this paper we revisit the problem considered by \cite{DavisNorman:90, ShreveSoner:94,ChoiSirbuZitkovic:13,HerczeghProkaj:15}. The fundamental difference between this paper and Choi et al~\cite{ChoiSirbuZitkovic:13} is that we take the primal approach. We show that the problem of constructing the value function can be transformed into finding the solution of a first-order free-boundary problem, subject to an integral condition, as in Choi et al~\cite{ChoiSirbuZitkovic:13}. The advance relative to \cite{ChoiSirbuZitkovic:13} is that it is much easier to understand the character of the solutions to our differential equation, when compared to that of Choi et al: for our solution the possible behaviours correspond to the possible shapes of a simple quadratic, whereas in Choi et al it is necessary to consider a phase-diagram in which the behaviours depend on a pair of ellipses and/or hyperbolas. Although the two free-boundary problems must be transformations of each other, our solution leads to a simpler problem. Most especially, we can give a direct interpretation of the free-boundary points as the sale and purchase boundaries of the no-transaction wedge and we can prove comparative statics for these boundaries. In the conclusion we will expand on this remark and give further evidence of the benefits of our approach. Nonetheless, many features of our characterising ODE are to be found in the characterising ODE of \cite{ChoiSirbuZitkovic:13}; in particular in both cases the ODE has a singular point, and for some parameter combinations, though not all, the solution we want passes through this singular point.

The remainder of this paper is structured as follows. In the next section we formulate the problem and discuss how to express the solution in terms of the solution of a first-order ODE. The Hamilton-Jacobi-Bellman equation is second-order, so the key step is an order-reduction in which we make the solution of the equation the independent variable. (This order reductiion technique has been used before in investment/sale problems by Evans et al~\cite{EvansHendersonHobson:08}. Choi et al~\cite[Section 3.3]{ChoiSirbuZitkovic:13} make a similar transformation, and this may be one reason why they make more progress than \cite{HerczeghProkaj:15}.) In Section~\ref{sec:sc} we do this in the case where it is never optimal for the agent to have negative cash wealth (equivalently never optimal to borrow against holdings of the risky asset), and hence we can use cash wealth as the denominator of an autonomous univariate process which is the ratio of wealth in the risky asset to cash wealth.

In Section~\ref{sec:gc}
we show how the arguments can be extended to the general case. We can no longer use cash wealth as the denominator in the definition of our autonomous univariate process. Instead we use paper wealth (where paper wealth is defined by assuming liquidation is possible with zero transaction costs) as the denominator and consider as key variable the ratio of wealth in the risky asset to paper wealth. At first sight, it looks as if this makes the order-reduction impossible, but inspired by the results of Section~\ref{sec:sc} we how the problem may be reduced to the identical free-boundary problem as in that section.

In Section~\ref{sec:cases} we show how the solution to the free-boundary problem depends on a simple quadratic, the co-efficients of which depend on the parameters of the problem. There are several cases depending only on the values and slopes of this quadratic at zero and one, and on the value of the quadratic at the turning point. We can give exact conditions which determine when the problem is well-posed. This main result is stated and proved in Section~\ref{sec:verification} and mirrors the main result of Choi et al~\cite{ChoiSirbuZitkovic:13}, but our formulation is an improvement in the sense that we cover an extra case and we give an algebraic expression for a quantity that Choi et al can only express as an integral.

In Section~\ref{sec:statics} we consider how the boundaries of the no-transaction wedge depend on the parameters. Analysis of this type seems to be new, and would be difficult under previous approaches. More especially, we show that if the drift is small, then the no-transaction wedge includes the Merton line, and the no-transaction wedge gets larger as transaction costs increase. However, if the drift increases further, then we may loose both the monotonicity property of the no-transaction wedge, and the property that the Merton line (corresponding to zero transaction costs) lies within the no-transaction region. Remarkably, although in general the locations of both the sale and purchase boundaries depend on the transaction costs on both sale and purchases, in some circumstances the sale boundary is independent of the transaction cost on purchases.

Section~\ref{sec:conc} concludes. Some results, including those on the solution of the ODE are given in an Appendix.

\section{Problem specification and a motivating special case}
\label{sec:sc}

Let $Y = (Y_t)_{t \geq 0}$ denote the price of a risky asset and suppose $Y$ is an exponential Brownian motion with drift $\mu$ and volatility $\sigma$; then $Y_t = Y_0 e^{\sigma B_t + (\mu - \sigma^2/2)t}$ where $B = (B_t)_{t \geq 0}$ is a Brownian motion. Let $C = (C_t)_{t \geq 0}$ denote the  consumption rate of the individual and let $\Theta_t$ denote the number of units of the risky asset held by the investor. We assume that $C$ is non-negative and progressively measurable and that
$\Theta$ is of finite variation; in particular $\Theta_t = \Theta_0 + \Phi_t - \Psi_t$ where $\Phi$ and $\Psi$ are increasing, adapted, c\`{a}dl\`{a}g processes with $\Phi_{0-} = \Psi_{0-} = 0$ representing purchases and sales of the risky asset respectively.

Suppose cash wealth is right-continuous and evolves according to
\begin{equation}
\label{eq:XdefnewS}
dX_t = - C_t dt - Y_t (1 + \lambda) d \Phi_t + Y_t (1-\gamma) d \Psi_t .
\end{equation}
Here $\lambda \in [0,\infty)$ represents the transaction cost paid on purchases and $\gamma\in [0,1)$ represents the transaction cost paid on sales. We assume $\lambda + \gamma > 0$, else we are in the case of no transaction costs.

We say that a wealth portfolio $(X_t,\Theta_t)$ is {\em solvent at time $t$} if
\[ X_t + (1-\gamma)\Theta_t^+ Y_t  -  (1+\lambda) \Theta_t^- Y_t > 0, \]
or equivalently if instantaneous liquidation of the risky position yields a cash wealth which is non-negative.  A consumption/investment strategy $(C, \Theta)$ is {\em solvent from time $t_0$} if the resulting wealth portfolio process $(X_t, \Theta_t)_{t \geq t_0}$ is solvent for each $t \geq t_0$. Write $\mathcal A = {\mathcal A}(x,y,\theta,t)$ for the set of strategies which are solvent from time $t$ when $(X_{t-}=x,Y_t=y, \Theta_{t-}=\theta)$.

The objective of the agent is to maximise
the discounted expected utility of consumption over the infinite horizon,
where the discount factor is $\beta$ and the utility function of the agent is assumed
to have constant relative risk aversion with risk aversion co-efficient $R \in (0,\infty) \setminus {1}$. The maximisation takes place over the set of consumption/investment strategies which are solvent from time zero.
In particular, the goal is to find
\begin{equation}
\underset{\left(C,\Theta\right)\in\mathcal{A}(x_{0},y_{0},\theta_{0},0)}
{\sup}\mathbb{E}
\left[\int_{0}^{\infty}e^{-\beta t}\frac{C_{t}^{1-R}}{1-R}dt\right].
\label{eq:objective}
\end{equation}

Since the set-up has a Markovian structure, we expect
the value function, optimal consumption and optimal portfolio strategy to be
functions of the current wealth portfolio of the agent and of the
price of the risky asset.
Let $V=V(x,y,\theta,t)$ be the forward starting value function
for the problem
so that
\begin{equation} V(x,y,\theta,t) =
\underset{\left(C,\Theta\right)\in\mathcal{A}(x,y,\theta,t)}
{\sup}\mathbb{E}\left[ \left. \int_{t}^{\infty}e^{-\beta
s}\frac{C_{s}^{1-R}}{1-R}ds \right| X_{t-} = x, Y_t=y, \Theta_{t-} =
\theta\right].
\label{eq:100b}
\end{equation}
The goal is to solve for the value function $V=V(x,y,\theta,t)$. Note that it is the value $y \theta$ of the holdings of the risky asset which is important rather than the price level and quantity individually, and in most circumstances $y$ and $\theta$ appear as the product $y \theta$. From the scalings of the problem we expect that we can write
\begin{equation}
\label{eq:Vdef} V(x,y,\theta,t) = e^{-\beta t} \frac{x^{1-R}}{1-R} g \left( \frac{y \theta}{x} \right),
\end{equation}
where the key variable is the ratio $z= y \theta/x$ of wealth held in the risky asset to cash wealth.

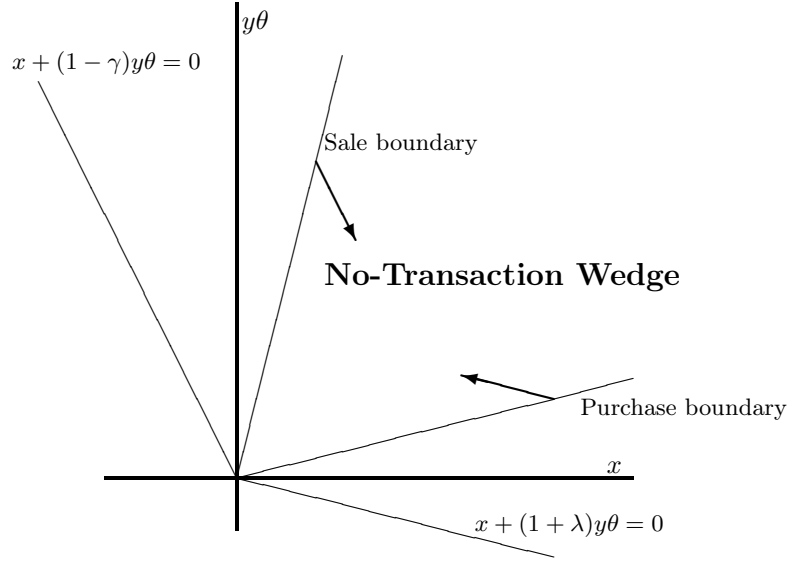
\begin{figure}
\begin{center}
\begin{picture}(200,200)(0,0)
\thicklines
\put (0,20){\line(1,0){200}}
\put (50,0){\line(0,1){200}}
\thinlines
\put (50,20){\line(4,-1){120}}
\put (50,20){\line(-1,2){75}}
\thinlines
\put (50,20){\line(4,1){150}}
\put (50,20){\line(1,4){40}}
\put (190,22){$x$}
\put (52,190){$y \theta$}
\put (140,0){{\small$x+(1+\lambda)y \theta=0$}}
\put (180,44){\small{Purchase boundary}}
\put (83,144){\small{Sale boundary}}
\put (83,93){{\Large{\bf No-Transaction Wedge}}}
\put (-35,175){\small{$x+(1-\gamma)y \theta=0$}}
\thicklines
\put (170,50){\vector(-4,1){35}}
\put (80,140){\vector(1,-2){15}}
\end{picture}
\end{center}
\caption{{\it The solvency and no-transaction regions}. The solvency region has boundaries given by the lines
$x+(1+\lambda)y \theta=0$ (for $x>0$ and $y \theta<0$) and $x+(1-\gamma)y \theta=0$ (for $x<0$ and $y \theta>0$). The no-transaction wedge is bounded by the sale and purchase boundaries. On and outside these boundaries, transactions are made to keep the process $(X_t, Y_t \Theta_t)$ inside the wedge.
The arrows represent the impact of transactions on the boundaries of the no-transaction wedge.}
\label{fig:wedge}
\end{figure}

The intuitive arguments of Constantinides and Magill~\cite{ConstantinidesMagill:76} and the concrete results of Davis and Norman~\cite{DavisNorman:90} lead us to expect that the no-transaction region will be a wedge. For the purposes of this section we suppose that this wedge is contained in the first quadrant of $(x,y \theta)$ space.
Define $Z = (Z_t)_{t \geq 0}$ by $Z_t = \Theta_t Y_t/X_t$. We expect that the agent trades to keep $Z$ in the interval $[z_*,z^*]$ where the pair of constants $ \{ z_*, z^* \}$ is to be determined. Then for initial values $y \theta > z^* x$ the optimal sale strategy includes an immediate sale to bring the ratio of risky wealth to cash wealth below $z^*$. Thus, if the initial portfolio $(X_{0-}=x,\Theta_{0-}=\theta)$ is such that $y_0 \theta > x z^*$ then we sell $\psi$ units of the risky asset, where $\psi = \frac{y_0 \theta - z^* x}{y_0(1+(1-\gamma)z^*)}$ so that
\[  \frac{y_0 \Theta_{0+}}{X_{0+}} \equiv \frac{y_0 \Theta_0}{X_0} = \frac{y_0 (\theta -\psi)}{x_0 + (1-\gamma)y_0\psi} =z^*. \]
This initial transaction at $t=0$ should not change the value function, and hence for $y \theta > x z^*$,
\[ x^{1-R} g \left( \frac{y \theta}{x} \right) = (x + (1-\gamma)y \psi)^{1-R} g(z^*) = \frac{(x+(1-\gamma)y\theta)^{1-R}}{(1 + (1-\gamma)z^*)^{1-R}} g(z^*),  \]
or equivalently $g(z) = (\frac{R}{\beta})^R A^*{(1 + (1-\gamma)z)^{1-R}}$ for $z>z^*$ where the constant $A^*$ is given by $A^* = (\frac{\beta}{R})^R \frac{g(z^*)}{(1 + (1-\gamma)z^*)^{1-R}}$.

If initial wealth is such that $y \theta < z_* x$ then the optimal strategy includes the immediate purchase of risky asset. We purchase $\phi$ units of $Y$ where $\phi = \frac{z_* x - y_0 \theta}{y_0(1 + (1+\lambda)z_*)}$ so that
\[ \frac{y_0 \Theta_{0+}}{X_{0+}} \equiv \frac{y_0 \Theta_0}{X_0} = \frac{y_0 (\theta +\phi)}{x_0 - (1 + \lambda)y_0\phi} =z_*. \]
Then $g(z) = (\frac{R}{\beta})^RA_*{(1 + (1+\lambda)z)^{1-R}}$ for $z \leq z^*$, where $A_* = (\frac{\beta}{R})^R \frac{g(z_*)}{(1 + (1+\lambda)z_*)^{1-R}}$.

Let $M = (M_t)_{t \geq 0}$ be given by
\[ M_t = \int_0^t e^{-\beta s} \frac{C_s^{1-R}}{1-R} ds + V(X_t, Y_t, \Theta_t, t). \]
Then we expect $M$ will be a supermartingale in general and a martingale under the optimal strategy.
Applying It\^{o}'s formula, and optimising over $C_t$ and $\Theta_t$ we obtain the Hamilton-Jacobi-Bellman equation which is a (second order, semi-linear) differential equation for $g$ in the no-transaction region: for $z_* < z < z^*$
\begin{equation}
\label{eq:h:gode}
0 = \frac{R}{1-R} \left( g(z) - \frac{z g'(z)}{1-R} \right)^{1-1/R} - \beta \frac{g(z)}{1-R} + \mu \frac{z g'(z)}{1-R} + \frac{\sigma^2}{2} \frac{z^2 g''(z)}{1-R} .
\end{equation}
Finally, we expect that there will be value matching and second-order smooth fit at the free boundary.

In analysing the problem our first goal is to solve \eqref{eq:h:gode}.  The equation can be simplified by setting $z=e^u$ and $h(u) = h(\ln z)= (\frac{\beta}{R})^R g(z)$. (If $z<0$ then we can construct a parallel argument based on $z =- e^{-u}$.)
Then $h$ solves a (second-order, non-linear) autonomous equation (with no $u$-dependence):
\begin{equation}
\label{eq:hode}
0 = \left( h - \frac{h'}{1-R} \right)^{1-1/R} - h + \left( \epsilon - \frac{\delta^2}{2} \right) h' + \frac{\delta^2}{2} h'' .
\end{equation}
where $\epsilon = \mu/\beta$ and $\delta^2 = \sigma^2/\beta$.
The order of this equation can be reduced by setting $\frac{dh}{du} = w(h)$ so that $\frac{d^2 h}{du^2} = h'' = w'(h) w(h)$. We find that
\begin{equation} \frac{\delta^2}{2} w(v) w'(v) -v  + \left(\epsilon -
\frac{\delta^2}{2}
\right) w(v) + \left( v - \frac{w(v)}{1-R} \right)^{1- 1/R}
= 0.
\label{eqn:wode}
\end{equation}

Various further transformations do not reduce the order of the problem, but rather simplify the problem in appearance, and add to our ability to interpret the solution. Set $W(h) = \frac{w(h)}{(1-R)h}$, let $N$ be inverse to $W$ so that $N = W^{-1}$ and finally set $n(q) = N(q)^{-1/R}(1-q)^{1-1/R}$. Then $n = n(q)$ solves
\begin{equation}
\label{eqn:node}
n'(q) = O(q,n(q))
\end{equation}
where
\begin{equation}
\label{eqn:O}
O(q,n) =
\frac{(1-R)}{R} \frac{n}{\left(1-q\right)}
-\frac{\delta^{2}}{2}\frac{\left(1-R\right)^{2}}{R}
\frac{q n}{\ell \left(q\right)-n}
 =
\frac{(1-R)}{R} \frac{n}{\left(1-q\right)} \frac{m(q)-n}{\ell(q) - n},
\end{equation}
and $m$ and $\ell$ are the quadratic functions
\begin{eqnarray}
\label{eq:mdef} m(q) & = & 1-\epsilon\left(1-R\right)q+\frac{\delta^{2}}{2}R\left(1-R\right)q^{2} \\
\label{eq:elldef} \ell(q) & = & 1 + \left( \frac{\delta^2}{2}-\epsilon \right) (1-R)q -
\frac{\delta^2}{2}(1-R)^2 q^2 =
m\left(q\right)+q\left(1-q\right) \frac{\delta^{2}}{2}\left(1-R\right).
\end{eqnarray}
It will turn out that the different solution regimes can be characterised by the quadratic $m$ and more especially, the derivatives of $m$ at 0 and 1 and the values of $m$ at 1 and at the turning point. To this end let $q_M = \frac{\epsilon}{\delta^2 R}$ be the location of the turning point, and let $m_M= m(q_M)$ be the value of $m$ at the turning point.

The advantage of switching to $n$ becomes apparent when we consider the solution outside the no-transaction region. For $z \leq z_*$, $g(z)= (\frac{R}{\beta})^R A_* (1 + (1 + \lambda)z)^{1-R}$ for $A_*$ to be determined. Then using the same transformations we find
$h(u) = (\frac{\beta}{R})^R g(e^u) = A_* (1 + (1 + \lambda)e^u)^{1-R}$ and
\begin{equation}
\label{eq:wode}
 w(h) = \frac{dh}{du} = (1-R) h \frac{(1+ \lambda)e^u}{1 + (1+ \lambda)e^u} = (1-R)h \frac{ (h/A_*)^{1/(1-R)} - 1}{(h/A_*)^{1/(1-R)}},
\end{equation}
so that $W(h) = 1 - (A_*/h)^{1/(1-R)}$, $N(q)=A_*(1-q)^{-(1-R)}$ and $n(q) = A_*^{-1/R}$ which is a constant. Similarly, on $z \geq z^*$ we have $n(q) = (A^*)^{-1/R}$.

Second order smooth fit of $g$ corresponds to first order smooth fit of $w$ (and $W$, $N$ and $n$). Hence we are looking for a solution $n$ and free boundaries $q_*$ and $q^*$ such that $n \in C^1$. Thus we require $n'=0$ at $q = q_*$ and $q= q^*$. However, the places in $(q,n)$ space where $n' = 0$ are exactly the points on the curve $(q, m(q))$. 
Hence, candidate solutions for the value function can be expressed in terms of solutions $n$ to \eqref{eqn:node} with boundary conditions $n(q_*) = m(q_*)$, $n(q^*)=m(q^*)$ for boundary points $\{q_*,q^*\}$ to be determined. Typically, there is a family of solutions to this problem, parameterised by the left-hand-endpoint $r$. Write $n_r = ( n_r(q) )_{q \geq r}$ for the solution to \eqref{eqn:node} started at the point $(r,m(r))$, and let
$\zeta(r) = \inf \{ q > r : (1-R)n_{r}(q) < (1-R)m(q) \}$. Clearly the solutions $n_r$ for different $r$ cannot cross, and $\zeta(r)$ is decreasing in $r$.

To fix ideas, suppose $0<R<1$, $0<\epsilon < \delta^2 R$ and $\epsilon^2 (1-R) < 2 \delta^2 R$. Then the quadratic $m$ is $U$-shaped with minimum at $q=q_M =\frac{\epsilon}{\delta^2 R} \in (0,1)$ and $m$ is non-negative everywhere. If $m(q) < n < \ell(q)$ then $n$ is decreasing at $q$; it can be shown that any solution $n_r$ started at $(r, m(r))$ with $0 < r <q_M$ lies strictly below the line joining $(0,m(0))$ with $(1,m(1))$.
With a little more work we can show that for $0 < r < q_M$, $\zeta(r) \in (q_M,1)$. Further, since the solutions $n_r$ cannot cross, $n_r(q)$ is decreasing in $r$ and $\zeta(r)$ is decreasing in $r$. See the first panel of Figure~\ref{fig:Genericnq1}.

Before describing how to choose the starting point $r$ corresponding to a given round-trip transaction cost $\xi$ it is useful to consider how important quantities of the value function and no-transaction region can be inferred immediately from (the correctly chosen) $n$. We have already seen that $A_* = n(q_*)^{-R}$ and $A^* = n(q^*)^{-R}$. Moreover, from \eqref{eq:wode} at $h_* = h(e^{u_*})$ where $u_* = \ln z_*$,
\[ q_* = W(h_*) = \frac{w(h_*)}{(1-R)h_*} = \frac{(1 + \lambda) z_*}{(1+(1 + \lambda) z_*)} = 1 - \frac{1}{(1+(1 + \lambda) z_*)}, \]
and similarly $q^*  = 1 -\frac{1}{1 + (1-\gamma)z^*}$.
In particular, we can infer the limits of the no-transaction region directly from the solution of the free boundary problem for $n$; $z_* = \frac{1}{1+\lambda} \frac{q_*}{1-q_*}$ and $z^* = \frac{1}{1-\gamma}\frac{q^*}{1-q^*}$.

Our objective is to solve the free-boundary problem:
\begin{quote} find $n$, $q_*$, $q^*$ such that $n$ is a nonnegative solution of \eqref{eqn:node} with boundary conditions $n(q_*) = m(q_*)$ and $n(q^*) = m(q^*)$.
\end{quote}
Typically there are many solutions to this problem and each solution corresponds to a different level of transaction costs. See the second panel of Figure~\ref{fig:Genericnq1}. The above equations provide the key to determining the solution we want. The solution corresponding to the round trip transaction cost $\xi = \frac{(1+\lambda)}{(1-\gamma)}-1 = \frac{\lambda +\gamma}{1-\gamma}$ must have
\begin{equation}
\label{eq:xicond}
\ln(1+\xi) = \ln(1+\lambda) - \ln(1-\gamma) = \ln \frac{z^*}{z_*} + \ln \left( \frac{(1-q^*)}{q^*} \frac{q_*}{(1-q_*)} \right).
\end{equation}
But $\frac{dh}{du} = w(h)$ so that $\ln \frac{z^*}{z_*} = u^* - u_* = \int_{h_*}^{h^*}\frac{dh}{w(h)}$ and then
\begin{equation}
\label{eq:xicond2}
\ln \left( \frac{z^*}{z_*} \right) = \int_{h_*}^{h^*} \frac{dh}{(1-R) h W(h)}
 = \int_{q_*}^{q^*} dq \frac{N'(q)}{(1-R)N(q) q} = \frac{\delta^2 (1-R)}{2} \int_{q_*}^{q^*} dq \frac{1}{( \ell(q) - n(q))}.
\end{equation}
Further,
\begin{equation}
\label{eq:xicond3}
\ln \left( \frac{q^*}{(1-q^*)} \right) - \ln \left( \frac{q_*}{(1-q_*)} \right) = \int_{q_*}^{q^*} \frac{dq}{q(1-q)} = \frac{\delta^2 (1-R)}{2} \int_{q_*}^{q^*} dq \left[ \frac{1}{\ell(q) - m(q)} \right] , \end{equation}
where we use $\ell(q)-m(q)= \frac{\delta^2}{2}(1-R)q(1-q)$. Hence
\begin{equation}
\label{eq:xicond4}
\ln \left(\frac{z^*}{z_*} \right) + \ln \left( \frac{(1-q^*)}{q^*} \frac{q_*}{(1-q_*)} \right) 
  =  \int_{q_*}^{q^*} dq  \frac{1}{q(1-q)} \left[ \frac{n(q)-m(q)}{\ell(q) - n(q)} \right] .
\end{equation} 

Define
\begin{equation}
 \label{eq:Lamdef}
 \Lambda(r) = \frac{\delta^2 (1-R)}{2} \int_{r}^{\zeta(r)} dq \left[ \frac{1}{\ell(q) - n_r(q)} - \frac{1}{\ell(q) - m(q)} \right] = \int_{r}^{\zeta(r)} dq  \frac{1}{q(1-q)} \frac{n_r(q)-m(q)}{\ell(q) - n_{r}(q)} \end{equation}
and set $\Sigma(q) = \exp( \Lambda(q)) - 1$.
Then, since $n_{r}(r) = m(r)$ and $n_{r}(\zeta(r)) = m(\zeta(r))$ we have from the first representation of $\Lambda(r)$,
\begin{equation}
\label{eq:Lamderiv}
 \frac{\partial \Lambda}{\partial r} =  \frac{\delta^2(1-R)}{2} \int_{r}^{\zeta(r)}
dq \frac{1}{(\ell(q) - n_{r}(q))^2}\frac{\partial n_{r}(q)}{\partial r} < 0 \end{equation}
where we use the fact that $n_{r}(q)$ is decreasing in $r$.
Further, we show in Lemma~\ref{lem:Lambda} in Section~\ref{sec:cases} below that
$\Lambda(0) = \infty$ and $\Lambda(q_M)=0$. Hence $\Sigma^{-1}$ is well defined on the domain $(0,\infty)$.
We set $q_* = \Sigma^{-1}(\xi)$ and $q^* = \zeta(q_*)$, and then $n_{q_*}$ solves \eqref{eqn:node} and \eqref{eq:xicond}.

\begin{figure}[!htbp]
	\captionsetup[subfigure]{width=0.4\textwidth}
	\centering
\subcaptionbox{Typical solutions $n_r$ together with $m$ and $\ell$.
	 \label{case1n}}{\includegraphics[scale =0.35] {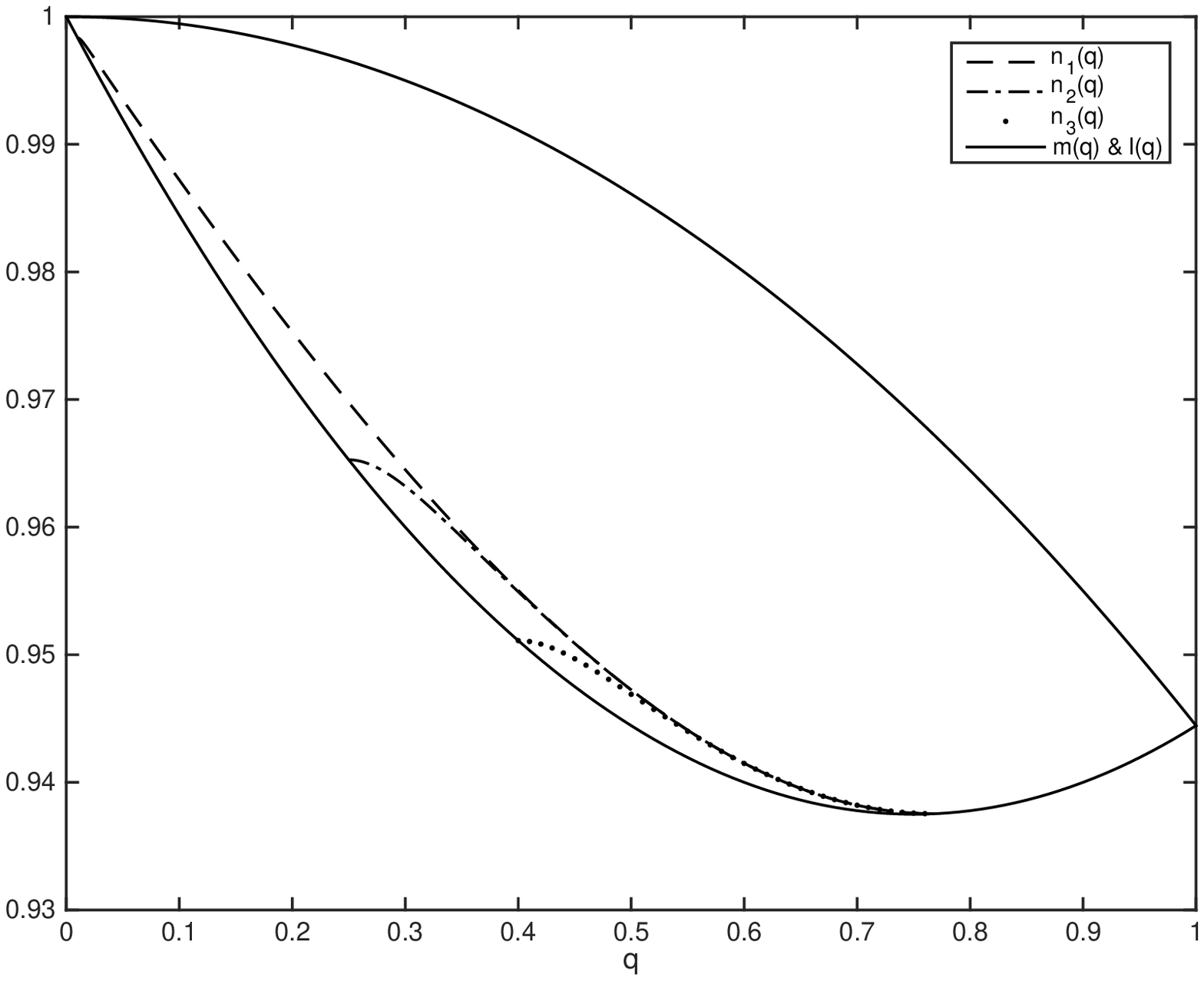}}	
\subcaptionbox{ $q_* = \Sigma^{-1}(\xi)$ and $q^* = \zeta (q_*)$
     \label{case1q}}{\includegraphics[scale =0.35] {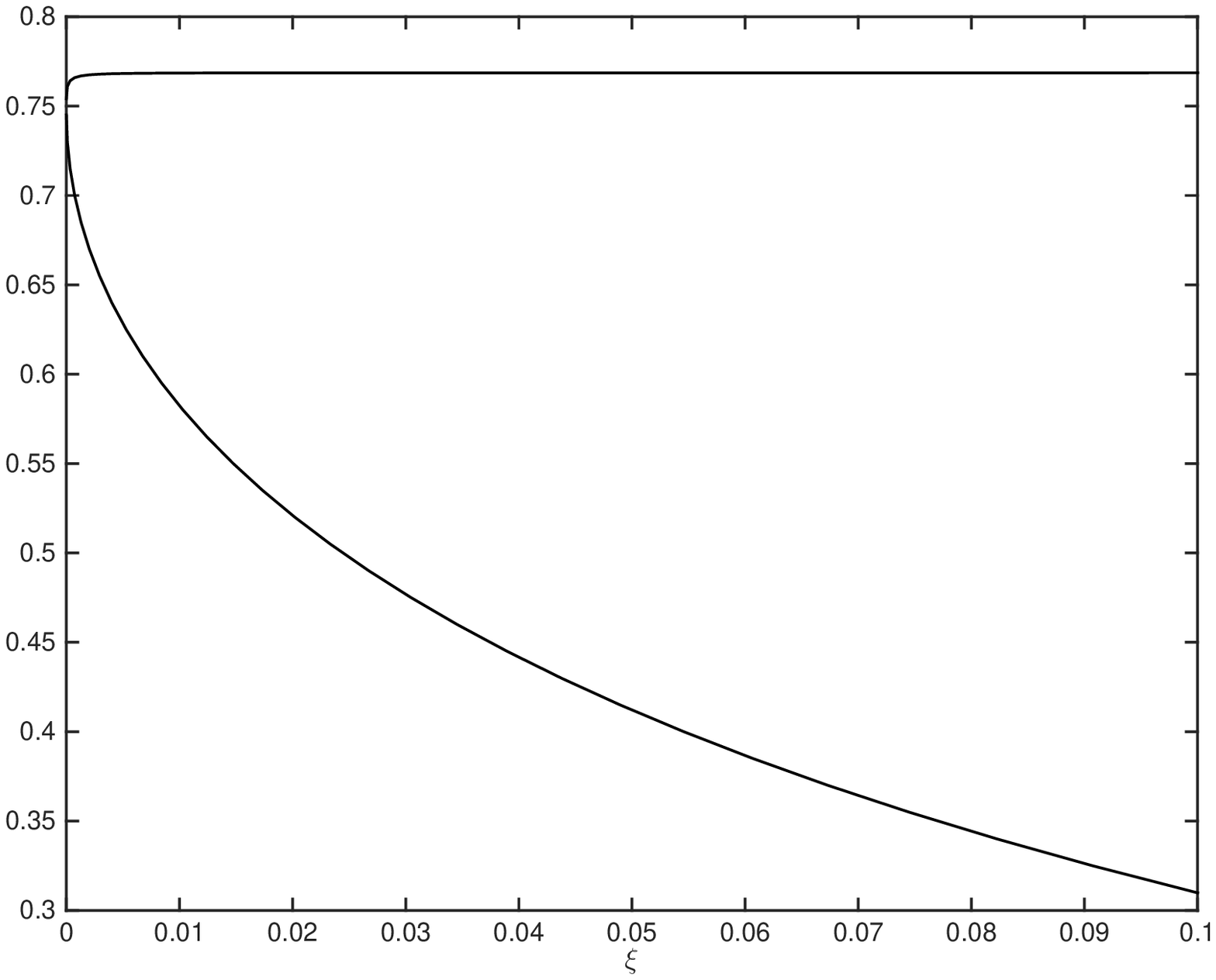}}
\caption{ {\small{
{\it Parameter values are $\epsilon = {1}/{2}$, $\delta = 1$ and $R=2/3$. Panel~\ref{case1n} shows some solutions $n_r(q)$. Panel~\ref{case1q} shows the boundaries of the no-transaction region as a function of the level of transaction cost. Note that the boundary $q^*$ corresponding to asset sales is insensitive to the level of transaction costs}.
}}}
\label{fig:Genericnq1}
\end{figure}

Our results are summarized in the following theorem, which is a special case of a more complete result given in Section~\ref{sec:verification}. At this stage we only have a candidate value function, but the verification argument that the candidate is indeed the value function is standard. Note that the candidate solution is $C^2$ throughout the solvency region.

\begin{thm}
\label{thm:special}
Suppose $R<1$ and $0 < \epsilon < \min \{ \delta^2 R , \delta \sqrt{\frac{2R}{1-R}} \}$.

Let $(n_r(\cdot))$ denote the family of solutions to \eqref{eqn:node} with initial values $n_r(r)= m(r)$. For any $\xi \in (0,\infty)$ there exists $\hat{r}$ such that $\Sigma(\hat{r}) = \xi = \frac{\lambda + \gamma}{1-\gamma}$, and
then the no-transaction wedge is $\frac{y \theta}{x} \in [z_*,z^*]$ where $z_* = \frac{1}{1+\lambda} \frac{\hat{r}}{1-\hat{r}}$ and $z^* = \frac{1}{1-\gamma}\frac{\zeta(\hat{r})}{1 - \zeta(\hat{r})}$.

Set $N(q)=n_r(q)^{-R}(1-q)^{R-1}$, $W=N^{-1}$ and $w(h) = (1-R)hW(h)$.
Set $h_* = N(\hat{r})$ and $h^*= N(\zeta(\hat{r}))$. The value function is given by \eqref{eq:Vdef} where in the no-transaction region $g(z) = (\frac{R}{\beta})^R h(e^u)$ and $u - \ln z_* = \int_{h_*}^h \frac{df}{w(f)}$; in the sale region $z>z^*$, $g(z) = h^* (1 + (1-\gamma)z)^{1-R} (1 + (1-\gamma)z^*)^{R-1}$; and in the purchase region $z<z_*$, $g(z) = h_* (1 + (1+\lambda)z)^{1-R} (1 + (1+\lambda)z_*)^{R-1}$.
\end{thm}

We close this section with a few remarks on the parametrisation of the problem, and on the form of the solution which hold both in the motivating special case of this section, and more generally.

Observe first that the assumption that the process $\Theta$ is finite variation can in fact be a conclusion, since if the agent follows a strategy which is not of bounded variation then transaction costs mean that she cannot keep her wealth process admissible.
Recall that $\xi =  \frac{1 + \lambda}{1 - \gamma} - 1 = \frac{\lambda + \gamma}{1- \gamma}$ where $\xi$ is the round-trip transaction cost. As we have seen it is $\xi$ which determines the nature of the solution to the problem, and not the individual transaction costs $\lambda$ and $\gamma$. Indeed, if we
define $\hat{Y}$ via $\hat{Y}_t = {Y_t}{(1-\gamma)}$ then \eqref{eq:XdefnewS} becomes
\begin{equation}
\label{eq:XdefnewSb}
dX_t = - C_t dt - \hat{Y}_t (1 + \xi) d \Phi_t + \hat{Y}_t d \Psi_t  .
\end{equation}
Thus, the problem with proportional transaction costs on both purchases and sales reduces to a problem with transaction cost $\xi$ on purchases only. Conversely, if we set $\tilde{Y_t} = Y_t(1+\lambda)$, then we have a problem in which the wealth process satisfies $dX_t = - C_t dt - \tilde{Y}_t d \Phi_t + \frac{1}{1+\xi} \tilde{Y}_t d \Psi_t$
corresponding to a problem with transaction cost $\frac{\xi}{(1 + \xi)}$ on sales only.

We have not included an interest rate term in the specification of the problem. However, the case of a constant interest rate can be reduced to the setting we describe by a simple switch to discounted units for both asset prices and consumption.

The quantities $q_M$ and $m_M$ have been defined as the location of the turning point of $m$ and the value of $m$ at that point. However, they have a direct interpretation in terms of the solution of the Merton problem with zero transaction costs. In the Merton problem with zero transaction costs the optimal strategy is to invest such that the ratio of wealth in the risky asset to total wealth (ie $\frac{\Theta_t Y_t}{X_t + \Theta_t Y_t}$) is kept equal to the constant $q_M$. Moreover, the value function under zero transaction costs $V=V(x,y,\theta,0)$ is equal to $\frac{(x+y\theta)^{1-R}}{1-R} ( \frac{\beta}{R})^R m_M^{-R}$.
In our present setting with $R<1$ it is clear that $n(q_*)=m(q_*) > n(q^*)=m(q^*) > m(q_M) =m_M$ and $A_* = n(q_*)^{-R} < A^* = n(q^*)^{-R} < m(q_M)^{-R} = m_M^{-R}$. Thus, the value function with non-zero transaction costs is bounded above by the value function with zero transaction costs.
Recall that in this section we are working under the hypothesis that $0<q_M<1$ and then $0 < q_* < q_M < q^* < 1$. If we set $z_M= \frac{q_M}{1-q_M}$ (with similar definitions for $z_*$ and $z^*$) then for $\xi>0$,
\[ z_* = \frac{1}{1+\lambda} \frac{q_*}{1-q_*} < \frac{q_*}{1-q_*} < \frac{q_M}{1-q_M} = z_M < \frac{q^*}{1-q^*} < \frac{1}{1-\gamma} \frac{q^*}{1-q^*} = z^*\] and the Merton line lies strictly within the no-transaction wedge. Note that this is not true in general, and there are parameter combinations for which the Merton line lies outside the no-transaction wedge.

\section{The general case}
\label{sec:gc}
In this section we want to develop and expand on the analysis of Section~\ref{sec:sc}, to include all parameter combinations, and not just those for which the
the no-transaction region is a wedge in the first quadrant.

One issue with the analysis of the previous section is that for some parameter values it may be that the line $(0,y \theta)$ lies inside the no-transaction wedge in $(x, y \theta)$-space. At $x=0$, $z = \frac{y \theta}{x}$ is ill-defined. For this reason in this section we consider an alternative parametrisation in which the key variable is the ratio of wealth in the risky asset to paper wealth, where paper wealth is calculated as the value of the portfolio under an assumption that holdings of the liquid asset can be sold or purchased at zero transaction cost. Note that the solvency requirement implies that paper wealth is non-negative.

Define $P = (P_t)_{t \geq 0}$ by $P_{t}= \frac{Y_t \Theta_t}{X_t + Y_t \Theta_t}$. The solvency requirement can be expressed as $- \frac{1}{\lambda} \leq P_t \leq \frac{1}{\gamma}$. Write
\begin{equation}
\label{eq:Vdef2} V(x,y,\theta,t) = e^{-\beta t} \frac{(x+ y \theta)^{1-R}}{1-R} \left(\frac{R}{\beta}\right)^R G \left( \frac{y \theta}{x+y \theta} \right) ,
\end{equation}
and consider
\begin{equation}
\label{eq:Mdef}
  M_t := \int_0^t e^{-\beta s} \frac{C_s^{1-R}}{1-R} ds + V(X_t, Y_t, \Theta_t, t) .
  \end{equation}
Applying It\^{o}'s formula we find (subscripts $x,y,\theta$ denote space derivatives and $\dot{V}$ denotes a time derivative)
\begin{eqnarray*} dM_t & = & \frac{C_t^{1-R}}{1-R} e^{-\beta t}dt + \dot{V} dt + V_x dX_t + V_y dY_t + V_\theta d\Theta_t + \frac{1}{2} V_{yy} d[Y]_t  \\
 & = & \left\{ \frac{C_t^{1-R}}{1-R} e^{-\beta t} - V_x  C_t \right\}
 dt  +[V_\theta - V_x(1+\lambda)Y_t] d\Phi_t + [V_x(1-\gamma)Y_t - V_\theta] d \Psi_t \\
     && \hspace{5mm} + e^{-\beta t} \frac{(X_t + Y_t \theta_t)^{1-R}}{1-R} \left( \frac{R}{\beta} \right)^R LG(P_t) dt + \sigma Y_t V_y dB_t
\end{eqnarray*}
where for $H=H(p)$
\[ LH = - \beta H + \mu \left[ (1-R) p H + p(1-p)H' \right]
      + \frac{\sigma^2}{2} \left[ -R(1-R) p^2 H - 2R p^2(1-p)H' + p^2 (1-p)^2 H'' \right]. \]

Since $M$ is a martingale under the optimal strategy and a super-martingale otherwise,
maximising over $C_t$ we find $C_t = e^{-\frac{\beta}{R}t} V_x^{-1/R} = \frac{\beta}{R}(X_t + Y_t \theta_t) \left[ G(P_t) - \frac{P_t G'(P_t)}{1-R} \right]^{-1/R}$. Since consumption is non-negative we must have that $G(p) > \frac{p G'(p)}{1-R}$.
Further, if $d \Phi_t>0$ then $V_\theta = (1+\lambda)yV_x$. Hence, if $d \Phi_t>0$ then $(1-R)G(P_t) + (1-P_t)G'(P_t) = (1 + \lambda)[(1-R)G(P_t) - P_tG'(P_t)]$ or equivalently
\[ -\lambda(1-R)G(P_t) + ( 1 + \lambda P_t) G'(P_t) = 0. \]
Similarly, if $d \Psi_t > 0$ then
\[ \gamma(1-R)G(P_t) + ( 1 - \gamma P_t) G'(P_t) = 0. \]
It follows that for $-\frac{1}{\lambda} \leq p \leq p_*$ we have $G(p) = \left( \frac{1+ \lambda p}{1 + \lambda p_*} \right)^{1-R} G(p_*)$ and for $p^* \leq p \leq \frac{1}{\gamma}$ we have
$G(p) = \left( \frac{1 - \gamma p}{1 -\gamma p^*} \right)^{1-R} G(p^*)$. (Alternatively we can derive the form of $G$ outside the no-transaction wedge by arguing that the optimal strategy includes an immediate transaction to move the wealth portfolio to the boundary of the no-transaction wedge, as in the previous section.)

Substituting for the optimal consumption we find that in the continuation region, $p_* < p < p^*$,
\begin{eqnarray}
0  & = & \beta \left[ G(p ) - \frac{p G'(p)}{1-R} \right]^{1-1/R} + LG(p) \nonumber \\
\label{eqn:Gode}   & = & \beta \left[ G(p ) - \frac{p G'(p)}{1-R} \right]^{1-1/R} - \beta{G}(p) + \mu \left[ (1-R) p G(p) + p(1-p)G'(p) \right] \\
 \nonumber && \hspace{10mm} + \frac{\sigma^2}{2} \left[ -R(1-R) p^2 G(p) - 2R p^2(1-p)G'(p) + p^2 (1-p)^2 G''(p) \right].
\end{eqnarray}
Now we can see the merit of the factor $\left(\frac{R}{\beta}\right)^R$ in the definition of $V$: we can divide through \eqref{eqn:Gode} by $\beta$ to reduce the problem to one expressed in the dimensionless quantities $\epsilon = \frac{\mu}{\beta}$ and $\delta^2 = \frac{\sigma^2}{\beta}$. (This completes the parameter reduction; the original parameters $\mu$, $\sigma$, $\beta$, $R$, $\lambda$, $\gamma$ have been replaced by $\epsilon$, $\delta$, $R$ and $\xi$.)

The parametrisation in terms of $P$ allows us to consider problems in which the no-transaction wedge lies outside the first quadrant. However, it has come at some expense: \eqref{eqn:Gode} appears considerably more complicated than \eqref{eq:h:gode} or \eqref{eq:hode}. However, inspired by the analysis of Section~\ref{sec:sc} we can make a transformation in which we recover the same first-order differential equation.

Set $h(p) = \sgn(p(1-p)) |1-p|^{R-1} G(p)$ and define $w(h) = p(1-p) \frac{dh}{dp}$. (The factor $p(1-p)$ is motivated by the identity $p = \frac{z}{1+z} = \frac{e^u}{1+e^u}$ whence $\frac{dp}{du} = p(1-p)$.) Then away from 0 and 1,
\[ \frac{dh}{dp} = \sgn(p(1-p)) |1-p|^{R-1} \left[ G'(p) + (1-R)\frac{G(p)}{1-p} \right] \]
and
\begin{equation}
\label{eq:wh}
 w(h) = \sgn(p(1-p)) |1-p|^{R-1} \left[ p(1-p)G'(p) + (1-R) p G(p) \right].
 \end{equation}
Moreover,
\[ w(h) \frac{d}{dh} w(h) = p(1-p)\frac{dh}{dp} \frac{d}{dh} w(h) = p(1-p)\frac{d}{dp} w(h)  \]
and on differentiating \eqref{eq:wh} we find for $p \notin \{0,1\}$
\begin{eqnarray*}
\lefteqn{ w(h) w'(h) } \\
& = & |1-p|^{R-1} \sgn(p(1-p)) \left[ p^2(1-p)^2 G''(p) + p(1-p)(1-2Rp)G'(p) + (1-R)p(1-Rp)G(p) \right]. \end{eqnarray*}
Then
\[ \left[ -R(1-R) p^2 G(p) - 2R p^2(1-p)G'(p) + p^2 (1-p)^2 G''(p) \right] = w(h)[w'(h)-1]|1-p|^{1-R}\sgn(p(1-p)). \]
Also using \eqref{eq:wh} we have
\[ G'(p) = \frac{w(h)}{|p| |1-p|^R} - (1-R) \frac{G(p)}{1-p} \]
and it follows that 
\[ G(p) - \frac{p G'(p)}{1-R} = |1-p|^{-R} \sgn(p) h(p) \left( 1 - \frac{w(h)}{(1-R)h} \right). \]
Since consumption must be non-negative this expression must be positive so we can write it as $G(p) - \frac{p G'(p)}{1-R} = |1-p|^{-R} |h| | 1 - \frac{w(h)}{(1-R)h}|$ and then
\[ \left( G(p) - \frac{p G'(p)}{1-R} \right)^{1-1/R} = |1-p|^{1-R} |h|^{1-1/R} \left| 1 - \frac{w(h)}{(1-R)h}\right|^{1-1/R} .  \]
Cancelling factors of $|1-p|^{1-R}$ and dividing by $\sgn(p(1-p)) = \sgn(h)$, Equation \eqref{eqn:Gode} becomes
\[ 0 =  h |h|^{-1/R} \left| 1 - \frac{w(h)}{(1-R)h} \right|^{1-1/R} - h + \epsilon w(h) + \frac{\delta^2}{2} w(h) \left[ w'(h) - 1 \right] , \]
and with $w(h) = (1-R)h W(h)$,
\[ \frac{\delta^2}{2}(1-R)^2 h W'(h) W(h) = - |h|^{-1/R} \left| 1 - W(h) \right|^{1-1/R} + \ell(W(h)). \]
Then setting $N= W^{-1}$ we find
\[ \frac{1}{N(q)} \frac{dN(q)}{dq} = \frac{\delta^2}{2}(1-R)^2 \frac{q}{\ell(q) - |N(q)|^{-1/R} |1-q|^{1-1/R}} .\]
Finally set $n(q) = |N(q)|^{-1/R} |1-q|^{1-1/R}$. Then $n>0$ and
\[ \frac{n'(q)}{n(q)} = \frac{1-R}{R(1-q)} - \frac{1}{R} \frac{N'(q)}{N(q)}. \]
In particular, $n$ solves $n' = O(q,n)$ where $O$ is as given by \eqref{eqn:O} for all values of $q \in [q_*,q^*]$ (except perhaps at the singular points $q=0$ and $q=1$).

Consider now the boundary conditions. For $-\frac{1}{\lambda} \leq p \leq p_*$ we have $G(p) = A_*(1+\lambda p)^{1-R}$. Then $h(p) = \sgn(p(1-p)) |1-p|^{R-1} A_*(1+\lambda p)^{1-R}$ and
\[ h'(p)= (1-R) h(p) \left[ \frac{1}{1-p} + \frac{\lambda}{1+\lambda p} \right] = (1-R) h(p) \frac{1+\lambda}{(1-p)(1+\lambda p)}. \]
It follows that $W(h) = \frac{(1+\lambda)p}{(1+\lambda p)}$; then $|1 - W(h)| = \frac{|1-p|}{1 + \lambda p} = (\frac{A_*}{|h|})^{1/(1-R)}$.
Writing $q=W(h)$ and $h = N(q)$ for $N=W^{-1}$ we have
\[ n(q) = |N(q)|^{-1/R} |1-q|^{1-1/R} = A_*^{-1/R}. \]
Note that $q = W(h) = \frac{(1+\lambda)p}{(1+\lambda p)}$ can be rewritten as
\begin{equation}
\label{eq:qpTR}
\frac{q}{1-q} = (1+ \lambda) \frac{p}{1-p}
\end{equation}
which is valid for $-\frac{1}{\lambda}<  p \leq p_*$ or equivalently
$- \infty < q < q_* = \frac{(1+\lambda)p_*}{(1+\lambda p_*)}$. A similar analysis gives
$n(q)= (A^*)^{-1/R}$ for $q \in [q^*,\infty)$ where $q^* = \frac{(1-\gamma)p^*}{(1-\gamma p^*)}$.
Thus, the condition of continuity of $n'$ at the free boundaries is equivalent to $n'=0$, which in turn means that candidate locations of the boundary can be identified with $O(q,n(q))=0$ or equivalently $n(q)=m(q)$.

Note that $q>1$ is equivalent to $p>1$ and that each of these conditions corresponds to the case of leverage (where the agent borrows to finance the position in the risky asset). Similarly $q<0$ is equivalent to $p<0$. These conditions corresponds to a short position in the risky asset.

From \eqref{eq:qpTR} at $(q_*,p_*)$ and the similar condition $\frac{q}{1-q} = (1 - \gamma) \frac{p}{1-p}$ at $(q^*,p^*)$ we have
\[ 1 + \zeta = \frac{1+\lambda}{1 -\gamma} = \frac{p^*}{1-p^*} \frac{1-p_*}{p_*} \frac{q_*}{1-q_*} \frac{1-q^*}{q^*} \]
and hence
\begin{equation}
\label{eq:4int}
\ln(1+\zeta) = \int_{p_*}^{p^*} \frac{dp}{p(1-p)} - \int_{q_*}^{q^*} \frac{dq}{q(1-q)} = \int_{h_*}^{h^*} \frac{dh}{w(h)} - \int_{q_*}^{q^*} \frac{dq}{q(1-q)} . \end{equation}
Following the same steps as in \eqref{eq:xicond}-\eqref{eq:xicond4} we conclude that the solution we want must satisfy
\begin{equation}
\label{eq:xidefG} \ln(1+\zeta) = \int_{q_*}^{q^*} dq  \frac{1}{q(1-q)} \frac{n(q)-m(q)}{\ell(q) - n(q)}. \end{equation}
Note that if $q_*<1<q^*$ then each of the integrals in \eqref{eq:4int} is over a domain which includes a singularity and hence the integral is not well-defined.
But, the integral in \eqref{eq:xidefG} is well defined since as we shall show $n(1)=m(1)$ and $n'(1)=m'(1)$, so that the integrand in \eqref{eq:xidefG} may be made bounded and
continuous at $q=1$.

The programme for constructing the value function is as before. Construct a family of solutions $n_r(\cdot)$ to $n'=O(q,n)$ parameterised by the initial value $n_r(r)=m(r)$. From this family choose the solution for which \eqref{eq:xidefG} is satisfied and let $N$, $W$ and $w$ be defined from the resulting $n_r$.
Then the candidate value function can be obtained by integrating $\frac{1}{w(h)}$ over the no-transaction region.

\section{Parameter regimes and possible behaviours}
\label{sec:cases}

We have shown above that the problem of constructing a candidate solution for the value function can be reduced to constructing a solution to a first order ordinary differential equation which starts and ends on a simple curve.

There are many cases to consider, each corresponding to different parameter regimes. However, the key point which we wish to emphasise is that the different cases can be distinguished by considering the behaviour of the function $m$ which is a simple quadratic.

We list ten cases. The different cases depend on the signs of the quantities $R-1$, $(1-R)m'(0)$, $m(1)$, $(1-R)m'(1)$ and $m_M$. Not all of the $2^5=32$ combinations are possible, and not all lead to different behaviours. (For instance, if $R>1$ then necessarily $m_M \geq 1 >0$; if $R>1$ and $m'(0)>0$, then the behaviour of the solution does not depend on the sign of $m(1)$ or the sign of $m'(1)$.)  In general, we do not analyse in detail the boundary cases, such as $R=1$ (which can be investigated using similar techniques, but would require a separate analysis) or $m'(0)=0$ etc (which can be understood as an appropriate limiting case). This is not because these cases are in any way difficult, but rather that it is simple to decide what should happen from the arguments we give below, and they bring no new insights.

Compared with the analysis in Section~\ref{sec:sc} the new cases bring new phenomena. First, we may find that the problem is ill-posed. Second, the problem may be ill-posed for low transaction costs, but have a finite solution for higher transaction costs. Third, the no-transaction wedge may lie in the second or fourth quadrants of $(x,y \theta)$-space, or may intersect both the first and second quadrants. In this case the find an ex-ante remarkable phenomena --- the boundary of the no-transaction region at which sales of the risky asset occur does not change as the level of transaction cost on purchases changes. (Ex-post, there is a simple explanation). The phenomena and associated cases are listed in Table~\ref{tab:list3}.

\begin{table}[htbp]
\begin{center}
\begin{tabular}{|l|c|c|c|c|}
\hline
                   &  & \multicolumn{3}{c|}{Location of NT wedge} \\ \cline{3-5}
                   & R & 1st Quadrant & 1st \& 2nd Quadrant  & 4th Quadrant   \\ \hline
 Unconditionally well-posed & $R < 1$ & {\em 1AbIIii} & {\em 1AbIii} & {\em 1Bii } \\ \cline{2-5}
                            & $R > 1$ & {\em 2AII} & {\em 2AI } & {\em 2B}  \\  \hline
 Conditionally well-posed   & $R < 1$ & {\em 1AbIIi} & {\em 1AbIi} & {\em 1Bi}  \\ \hline \cline{3-5}
 Unconditionally ill-posed  & $R < 1$ & \multicolumn{3}{c|}{ {\em 1Aa} } \\
\hline
\end{tabular}
\caption{The different cases, arranged phenomenologically}
\label{tab:list3} \end{center}
\end{table}

We distinguish the cases as follows, based on the behaviour of $m$. Call $R<1$ {\em Case 1} and $R>1$ {\em Case 2}. Call $(1-R)m'(0)<0$ {\em Case A} and $(1-R)m'(0)>0$ {\em Case B}. In {\em Case 1A} only, call $m(1)<0$ {\em Case a} and $m(1)>0$ {\em Case $b$}. (Note that in {\em Case 1B}, we must have $m(1)>m(0)=1>0$, so the sign of $m(1)$ is determined; in {\em Case 2}, it turns out that the sign of $m(1)$ is not important.) In {\em Cases 1Ab} and {\em 2A}, call $(1-R)m'(1)<0$ {\em Case I} and call $(1-R)m'(1)>0$ {\em Case II}. Finally, in {\em Cases 1AbI}, {\em 1AbII} and {\em 1B} call $m_M<0$ {\em Case i} and $m_M>0$ {\em Case ii}.

Table~\ref{tab:list} lists all the different cases. Where there is no entry in a cell, it means that the same analysis covers both possible cases for that value.
(So, for example, in {\em Case 1Aa}, the form of the solution does not depend on the sign of $m'(1)$.) Where the entry in the cell is +ve or -ve it means that the sign of the cell is determined by the signs of previous cells in the row. (So, for example, if $R>1$ then $m_M>0$ necessarily.)
It follows from exhaustion that all possible parameter combinations are included in one of the rows, except the boundary cases for which $\epsilon \in \{ - \delta \sqrt{\frac{2R}{1-R}}, 0, \delta^2 R, \delta \sqrt{\frac{2R}{1-R}}, \frac{1}{1-R} + \frac{\delta^2 R}{2} \}$. Note that necessarily $\delta \sqrt{\frac{2R}{1-R}} \leq \frac{1}{1-R} + \frac{\delta^2 R}{2}$, but that any ordering between $\delta^2 R$ and these two quantities is possible.

\begin{table}[htbp]
\begin{center}
\begin{tabular}{|l|c|cccc|c|c|c|c|}
\hline
                   &  R     & $m'(0)$   & $m(1)$   & $m'(1)$   &  $m_M$  & Range of values of $\epsilon$   & $\epsilon$ & $\delta$ & R  \\ \hline
{\em Case 1AbIIii} & $<1$  & $<0$ & $>0$ & $>0$ & $>0$ & $0 < \epsilon < \min \{ \delta^2 R , \delta \sqrt{\frac{2R}{1-R}} \}$ &  1/2 & 1 & 2/3 \\
{\em Case 1Aa}     & $<1$  & $<0$ & $<0$ &      & -ve  & $\frac{1}{1-R} + \frac{\delta^2 R}{2} < \epsilon$ & 35/2 & 6 & 2/3 \\
{\em Case 1AbIIi}  & $<1$  & $<0$ & $>0$ & $>0$ & $<0$ & $\delta \sqrt{\frac{2R}{1-R}} < \epsilon < \min \{ \delta^2 R , \frac{1}{1-R} + \frac{\delta^2 R}{2} \}$ & 27/2 & 6 & 2/3 \\
{\em Case 2AII}    & $>1$  & $>0$ &      & $<0$ & +ve  & $0 < \epsilon < \delta^2 R$ & 1 & 1 & 2  \\
{\em Case 1AbIii}  & $<1$  & $<0$ & $>0$ & $<0$ & $>0$ & $\delta^2 R < \epsilon <  \delta \sqrt{\frac{2R}{1-R}}$ & 3/2 & 1 & 2/3 \\
{\em Case 1AbIi}   & $<1$  & $<0$ & $>0$ & $<0$ & $<0$ & $\max \{ \delta^2 R , \delta \sqrt{\frac{2R}{1-R}} \} < \epsilon <  \frac{1}{1-R} + \frac{\delta^2 R}{2} \}$ & 13/4 & 3/2 & 2/3 \\
{\em Case 2AI}     & $>1$  & $>0$ &      & $>0$ & +ve  & $\delta^2 R < \epsilon$  & 5/2 & 1 & 2\\
{\em Case 1Bii}    & $<1$  & $>0$ &      &      & $>0$ & $- \delta \sqrt{\frac{2R}{1-R}} < \epsilon < 0$  & -1 & 1 & 2/3 \\
{\em Case 1Bi}     & $<1$  & $>0$ &      &      & $<0$ & $\epsilon < - \delta \sqrt{\frac{2R}{1-R}}$  & - 3 & 1 & 2/3 \\
{\em Case 2B}      & $>1$  & $<0$ &      &      & +ve  & $\epsilon < 0$ & -1 & 1 & 2 \\
\hline
\end{tabular}
\caption{The different cases, and associated parameter values. The last three columns refer to parameter values used in numerical examples.}
\label{tab:list} \end{center}
\end{table}

The distinction between {\em Cases A} and {\em B} is that in the former case $Y$ is a depreciating asset, in the latter case $Y$ has positive drift. Then in {\em Case A} the no-transaction wedge is contained in the upper-half-plane, and in {\em Case B} it is a subset of the fourth quadrant.

The distinction between {\em Cases a} and {\em b} is that in {\em Case a} the problem is ill-posed and the value function is infinite. Note that the problem can only be ill-posed if $R<1$.

The distinction between {\em Cases I} and {\em II} is that in the former case the solution $n$ may pass through the singular point $(1, m(1))$. Then in {\em Case I} the no-transaction wedge intersects the second quadrant (and may be a strict subset of the second quadrant) whereas in {\em Case II} the no-transaction wedge is contained in the first quadrant. In {\em Case I}, for large enough $\xi$, the value of $q^*$ does not depend on the round-trip transaction cost $\xi$.

Finally, the distinction between {\em Cases i} and {\em ii} is that in the latter case the problem has a solution for all transaction costs; in {\em Case i} we have $m_M<0$ and then the problem is ill-posed if the round-trip transaction cost is sufficiently small.

The descriptions of the various cases should be studied in parallel with
Figures~\ref{fig:Genericnq2}---\ref{fig:Genericnq10}
which provide a pictorial representation of the different cases. Proofs are given in the next section; in this section we describe and characterise the solution in terms of the behaviour of the quadratic $m$.


\subsection{Case 1AbIIii: $R<1$, $0 <\epsilon< \min \{ \delta^2 R , \delta \sqrt{\frac{2 R}{1-R}} \}$. }
We begin with considering Case 1AbIIii. This is the case we considered in Section~\ref{sec:sc}, and is the simplest case. We have that $m$ has a minimum at $q_M = \frac{\epsilon}{\delta^2 R} \in (0,1)$ and $m_M>0$.

The following result is completely intuitive given the definitions of $\zeta$ and $\Lambda$. A proof is given in the Appendix.

\begin{lem}
\label{lem:3.1}
For all starting points $r \in (0,q_M)$ we have $\zeta(r) \in (q_M, 1)$. Also $\zeta(q_M)=q_M$.

$\Lambda(q_M) = 0$, $\lim_{r \downarrow 0} \Lambda(r)=\infty$ and $\Lambda$ is continuous and strictly increasing.
\end{lem}

It follows from the Lemma that $\Sigma:(0,q_M] \mapsto [0,\infty)$ is onto and that for any $\xi \in [0,\infty)$ there exists a solution to the free-boundary problem: $n$ solves \eqref{eqn:node} subject to $n(q_*) =m(q_*)$, $n(q^*) =m(q^*)$, $\Sigma(q_*)=\xi$. Moreover the solution has the additional property that $0<q_* \leq q_M \leq q^* < 1$.

Let $\zeta(0) = \lim_{r \downarrow 0} \zeta(r) < 1$ and let
$n_0$ defined on $[0,\zeta(0)]$ be given by $n_0(q) = \lim_{r \downarrow 0} n_r(q)$.
Both $\zeta(0)$ and $n_0$ are well defined by the monotonicity of solutions to \eqref{eqn:node}.
It follows that however large the transaction costs, the no-transaction wedge is a strict subset of the first quadrant and is such that $[q_*,q^*] \subset (0, \zeta(0)) \subset (0,1)$.

\subsection{Case 1Aa: $R<1$, $ \frac{1}{1-R} + \frac{\delta^2 R}{2} <\epsilon$. }
In this case $m(1)<0$. Let $q_\pm$ be the roots of $m$ (with $q_-<q_+$) and let $p_\pm$ be the roots of
$\ell$ (with $p_- < p_+$). Note that $p_- <0<q_-<p_+<1<q_+$. It is clear from the presence of the factor $n$ in the numerator of $O$ that $n$ cannot hit zero before $\ell$ hits zero at $p_+$. See Figure~\ref{fig:Genericnq2}. On the other hand $n \leq \ell$ on $(0,p_+)$. Thus for any $r \in (0, q_-)$ we have $n_r(p_+)=0$. It follows that there is no solution to the free boundary problem. The optimal consumption/investment problem is ill-posed (in the sense that there is a strategy which generates infinite expected discounted utility) for {\em any} value for the round-trip transaction cost. (One strategy is at time zero to trade to a cash only position, and thereafter to keep $\Theta_t = 0$ and to consume at the constant rate $C_t = \frac{\beta}{R} X_t$ per unit time.)

Note that in this case we treat both $\epsilon < \delta^2 R$ and $\epsilon>\delta^2 R$ in the same fashion: the problem is ill-posed in both cases. Thus we do not distinguish between {\em Case I} and {\em Case II}. On the other hand, if $\epsilon > \frac{1}{1-R} + \frac{\delta^2 R}{2}$ then necessarily $ \epsilon > \delta \sqrt{\frac{2R}{1-R}}$. Thus {\em Case ii} cannot occur, and we must be in {\em Case i}.

\begin{figure}[!htbp]
	\captionsetup[subfigure]{width=0.4\textwidth}
	\centering
\subcaptionbox{Typical solutions $n_r$ together with $m$ and $\ell$.
	 \label{case2n}}{\includegraphics[scale =0.4] {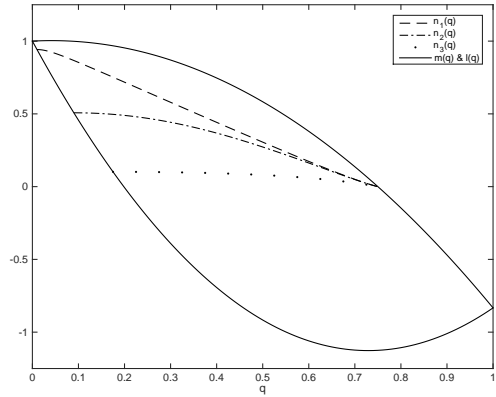}}	
\caption{ {\small{
{\it Case 1Aa. Parameter values are $\epsilon = 35/2$, $\delta = 6$ and $R=2/3$.
In this case all candidate solutions $n_r$ with $r \leq q_-$ hit zero at $p_+$. The problem is ill-posed}.
}}}
\label{fig:Genericnq2}
\end{figure}

\subsection{Case 1AbIIi: $R<1$, $\delta \sqrt{\frac{2R}{1-R}} <\epsilon < \min \{ \frac{1}{1-R} + \frac{\delta^2 R}{2}, \delta^2 R \}$. }
In this case $m(1)>0$ but the turning point of $m$ is at $q_M \in (0,1)$ and $m$ takes a negative value at the turning point. It follows that for these parameter values the problem with zero transaction costs is ill-posed. We argue that the problem remains ill-posed for small transaction costs. However, for large transaction costs the value function is finite. Moreover, we can identify the threshold value $\underline{\xi}$ of the round-trip transaction cost which lies at the boundary between the two regimes.

Let $q_\pm$ be the roots of $m$ as above. For $0<r<q_-$ we can define a family of non-negative solutions $n_r$ to \eqref{eqn:node}. Note $\lim_{r \uparrow q_-} \zeta(r) = q_+$, and $\lim_{r \uparrow q_-} n_r(q)=0$ on $[q_-,q_+]$. Then
\begin{equation}
\label{eq:undLamdef}
 \Lambda(q_-) = \int_{q_-}^{q_+} dq \frac{1}{q(1-q)} \frac{|m(q)|}{\ell(q)} .
\end{equation}
This integral can be evaluated (see Proposition~\ref{prop:integral} in the appendix) and we find
\begin{eqnarray*}
\Lambda(q_-) = \underline{\Lambda}
& := & - \ln \frac{q_+}{q_-}
-  \ln \frac{1-q_-}{1-q_+} \\
&& \hspace{5mm}
+ \frac{R}{1-R}\frac{(p_+ - q_+)(p_+ - q_-)}{p_+(p_+ - 1)(p_+ - p_-)} \ln \frac{p_+- q_-}{p_+ - q_+} \\
&& \hspace{5mm}
- \frac{R}{1-R} \frac{(q_+ - p_-)(q_- - p_-)}{p_-(1-p_-)(p_+ - p_-)} \ln \frac{q_+ - p_-}{q_- - p_-}
\end{eqnarray*}
where
\[ q_\pm =
\frac{\epsilon \pm \sqrt{\epsilon^2 - (\delta \sqrt{\frac{2R}{1-R}})^2}}{\delta^2 R};
\hspace{10mm}
p_\pm = \frac{\frac{1}{2}\delta^2 - \epsilon \pm \sqrt{2 \delta^2 + (\frac{1}{2} \delta^2 - \epsilon)^2}}{\delta^2(1-R)} . \]
Then, for $\xi > e^{\underline{\Lambda}}-1$ the optimal consumption/investment problem with transaction costs is well-posed, but for $\xi \leq e^{\underline{\Lambda}} - 1$ the problem is ill-posed, and there is a strategy which yields infinite expected utility from consumption. See Figure~\ref{fig:Genericnq3}.

\begin{figure}[!htbp]
	\captionsetup[subfigure]{width=0.4\textwidth}
	\centering
\subcaptionbox{Typical solutions $n_r$ together with $m$ and $\ell$.
	 \label{case3n}}{\includegraphics[scale =0.35] {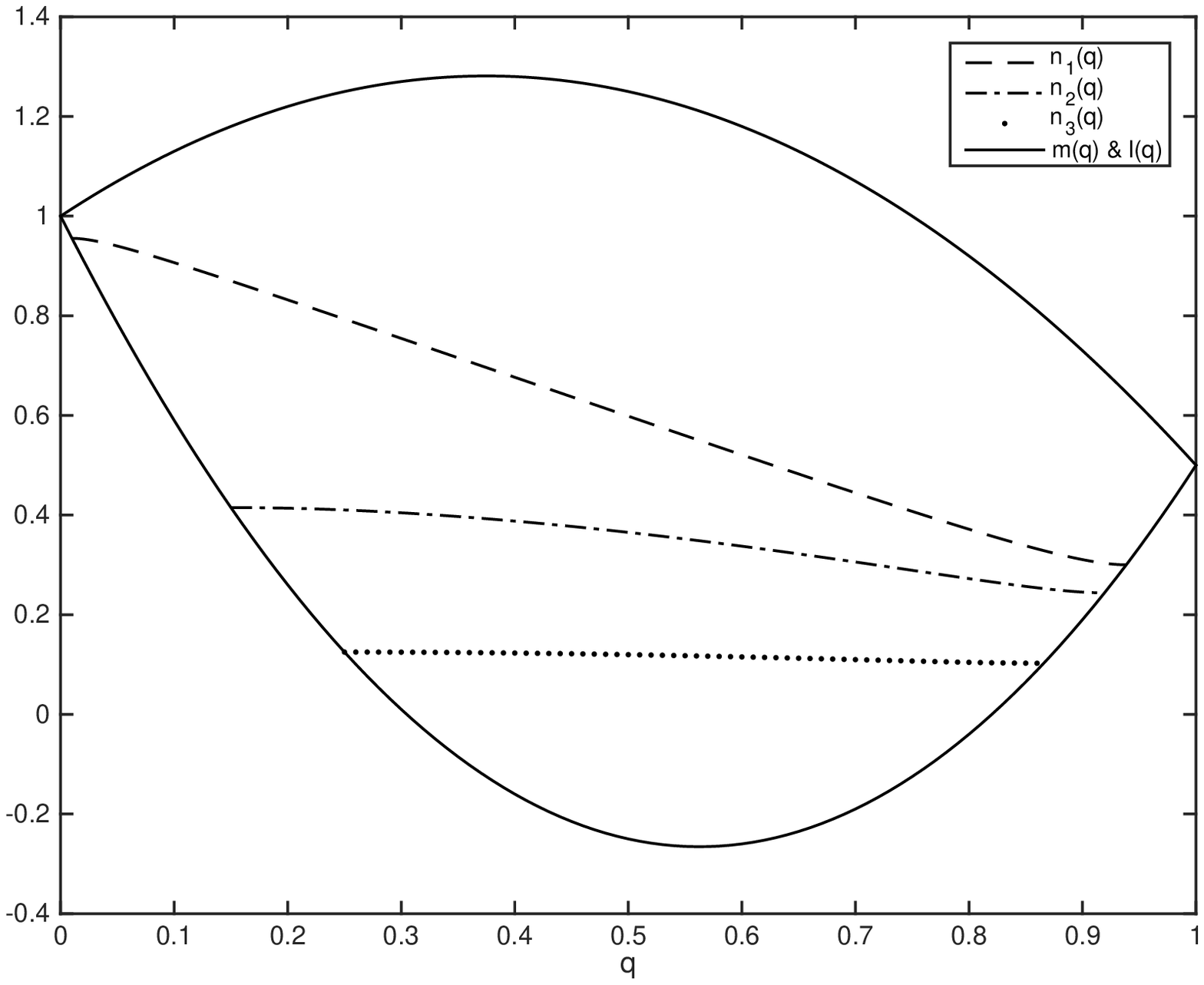}}	
\subcaptionbox{$q_* = \Sigma^{-1}(\xi)$ and $q^* = \zeta (q_*)$
     \label{case3q}}{\includegraphics[scale =0.35] {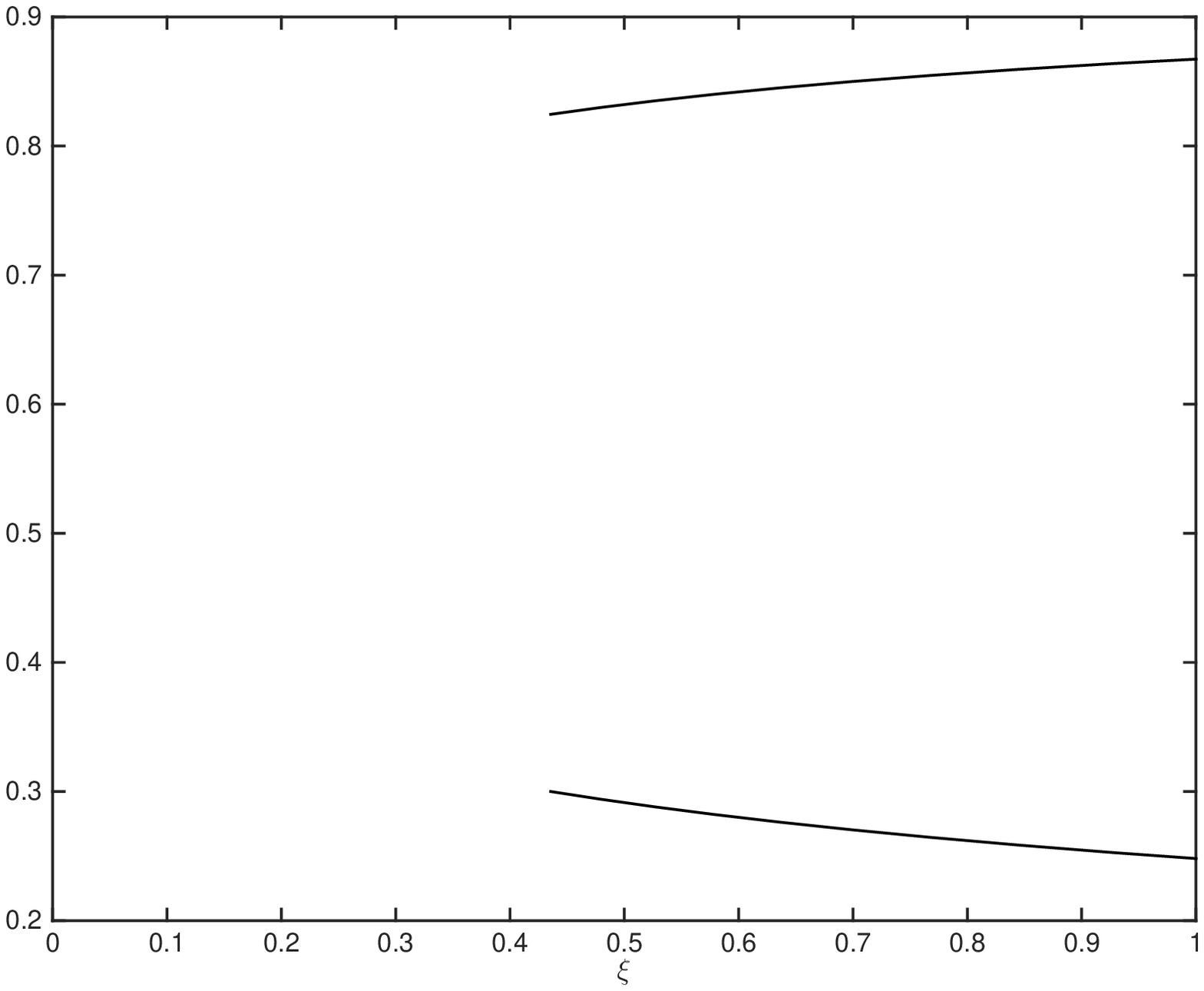}}
\caption{ {\small{
{\it Case 1AbIIi. Parameter values are $\epsilon = 27/2$, $\delta = 6$ and $R=2/3$. Note that if $\xi$ is too small then the problem is ill-posed. At the critical $\xi$, $\xi = e^{\underline{\Lambda}}-1$ we have $n_{q_*} = 0$ on $(q_*,q^*)$}.
}}}
\label{fig:Genericnq3}
\end{figure}

\subsection{Case 2AII: $R>1$, $0 < \epsilon < \delta^2 R$.}
In this case $m>\ell$ on $(0,1)$ and $n$ is increasing provided $q \in (0,1)$ and $\ell(q) < n(q) < m(q)$. See Figure~\ref{fig:Genericnq4}. The condition $0 < \epsilon < \delta^2 R$ ensures that $m$ has a turning point at $q_M \in (0,1)$. Using the same reasoning as in Lemma~\ref{lem:3.1}, as in {\em Case 1AbIIii} for $r \in (0,q_M)$ we must have $\zeta(r) \in (q_M, 1)$. Again as in {\em Case 1AbIIii} we can define $\zeta(0)$ and $n_0$ and we find $\zeta(0) < 1$. However, unlike for the case $R<1$, the value of $m(1)$ does not matter; since all solutions $n_r$ (with $r \in (0,q_M)$) are increasing and lie between $\ell$ and $m$, and since $m'(1)<0$, they must intersect $m$ before $1$ and they must stay positive. The problem cannot be ill-posed if $R>1$. (Of course, this is clear from the utility function; if $U(c) = \frac{c^{1-R}}{1-R}$ then $U$ is bounded above and the expected utility from consumption is also bounded above. Conversely a strategy for which $\Theta_t \equiv 0$, and which involves consuming a constant fraction of wealth per unit time thereafter (so that $C_t = \frac{\beta}{R} X_t$) yields a finite lower bound on the value function.)

Recall the final remarks in Section~\ref{sec:sc} relating $n$ to the value function of the zero-transaction cost Merton problem in the case $R<1$.
Now we have that $n(q_*) < n(q^*) < m_M$ and $A_* = n(q_*)^{-R} > A^* = n(q_*)^{-R} > m_M^{-R}$. Note that in this case the presence of the term $\frac{1}{1-R}$ in the value function means that increases in $n$ (or $G$) decrease the value function. Hence again the value function is bounded above by the value function with zero transaction costs.

\begin{figure}[!htbp]
	\captionsetup[subfigure]{width=0.4\textwidth}
	\centering
\subcaptionbox{Typical solutions $n_r$ together with $m$ and $\ell$.
	 \label{case4n}}{\includegraphics[scale =0.35] {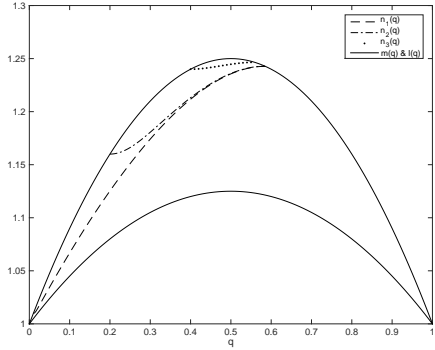}}	
\subcaptionbox{$q_* = \Sigma^{-1}(\xi)$ and $q^* = \zeta (q_*)$
     \label{case4q}}{\includegraphics[scale =0.35] {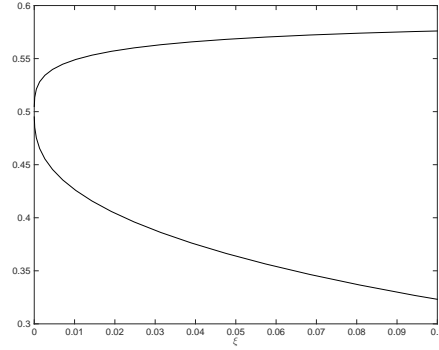}}
\caption{ {\small{
{\it Case 2AII. Parameter values are $\epsilon = 1$, $\delta = 1$ and $R=2$.
Since $R>1$ we now find $m > \ell$ over $(0,1)$ and the solutions we want satisfy $\ell < n < m$}.
}}}
\label{fig:Genericnq4}
\end{figure}

\subsection{Case 1AbIii: $R<1$, $\delta^2 R <\epsilon< \delta\sqrt{\frac{2R}{1-R}}$. }
In this case $m$ is positive everywhere, but $m'(1)<0$ so that $m$ is decreasing on $(0,1)$, and for $r<1$, $\zeta(r) > 1$. See Figure~\ref{fig:Genericnq5}.

Since $q_M>1$ the Merton line lies in the second quadrant. For small transaction costs the no-transaction wedge lies in the second quadrant and the agent has a leveraged position in the risky asset, ie borrows to finance the position in the risky asset. This is the case for which we can find a solution $n_r$ with $r>1$ such that $\Sigma(r)=q$.
For large transaction costs we have $0 < q_* < 1 < q_M <q^*$ and the no-transaction wedge intersects both quadrants in the upper-half-plane.
The threshold $\overline{\xi}$ between small and large transaction costs is not given as an algebraic function, but is $\overline{\xi}=e^{\Lambda(1)}-1$ where
\begin{equation}
\label{eq:Lambda1}
\Lambda(1) = \int_1^{\zeta(1)} dq \frac{1}{q(q-1)}\frac{n_1(q)-m(q)}{n_1(q)- \ell(q)}.
\end{equation}

In the case $\xi \geq \overline{\xi}$ we find that $q^*$ does not depend on $\xi$.

\begin{lem}
\label{lem:singular}
(i) $n_1(\cdot)$ is well defined. Further $\zeta(1)>q_M>1$ and $n'_1(1)=m'(1)<0$. \\
(ii) For $0<r<1$, $n_r(1) = m(1)$ and $n_r'(1)=m'(1)<0$.  \\
(iii) For $0<r<1<q<\zeta(1)$, $n_r(q) = n_1(q)$. In particular, for $0<r<1$, $\zeta(r)=\zeta(1)$.
\end{lem}

\begin{proof}
See appendix.
\end{proof}

The intuition behind these results is as follows.
Since $q_M>1$ and $m$ is decreasing on $(0,1)$, for $r \in (0,1)$ we have that $n_r$ cannot cross $m$ before $q=1$. Since we have $n_r(q) \leq \ell(q)$ on $(0,1)$ we must have that $n$ passes through the singular point $(1, m(1))$. At this point $n'(1)=m'(1)$. A solution for $n$ can be constructed beyond $q=1$, but since $n$ solves a first order equation, the solution does not depend in any way on the behaviour of $n$ to the left of $1$. Thus, if $r<1$, $\zeta(r)$ does not depend on $r$.

\begin{lem}
\label{lem:Lambda}
$\Lambda(r)$ is continuous at $r=1$. Further $\Lambda$ is strictly decreasing with $\Lambda(q_M)=0$ and $\lim_{r \downarrow 0}\Lambda(r)=\infty.$
\end{lem}

\begin{proof}
See appendix.
\end{proof}

Given the function $\Lambda$ we can define $\Sigma(r) = e^{\Lambda(r)}-1$ and then $\Sigma^{-1}$ is well-defined. We set $q_* = \Sigma^{-1}(\xi)$ and $q^* = \zeta(q_*)$ and derive the value function in the no-transaction wedge by integrating $n_{q_*}$ over the interval $[q_*,q^*]$.

The following corollary is an immediate consequence of the fact that for $\xi \geq \overline{\xi}$, $q^* = q^*(\xi) = \zeta(1)$ does not depend on $\xi$.

\begin{figure}[!htbp]
	\captionsetup[subfigure]{width=0.4\textwidth}
	\centering
\subcaptionbox{Typical solutions $n_r$ together with $m$ and $\ell$.
	 \label{case5n}}{\includegraphics[scale =0.35] {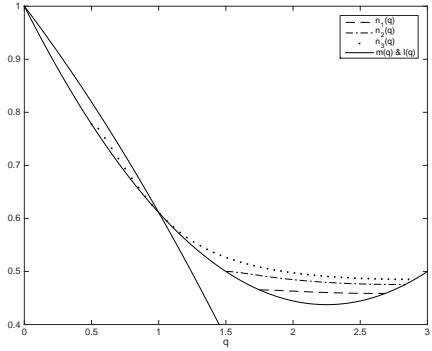}}	
\subcaptionbox{$q_* = \Sigma^{-1}(\xi)$ and $q^* = \zeta (q_*)$
     \label{case5q}}{\includegraphics[scale =0.35] {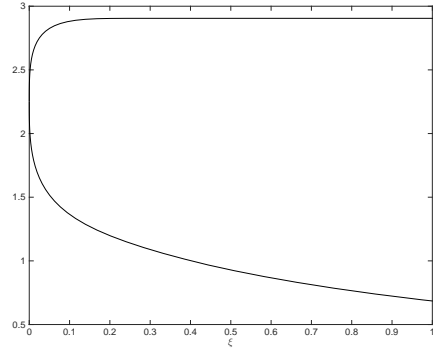}}
\caption{ {\small{
{\it Case 1AbIii. Parameter values are $\epsilon = 3/2$, $\delta = 1$ and $R=2/3$.
Since $m'(1)<0$, for each $r \in (0,1)$ we have $n_r(q)$ passes through the point $(1,m(1))$.
For large enough $\xi$, $q^*$ is constant}.
}}}
\label{fig:Genericnq5}
\end{figure}

\begin{cor}
\label{cor:surp}
For $\xi > \overline{\xi}$ the no transaction wedge intersects both the first and second quadrants. If the transaction cost on sales is held fixed, then provided $\xi > \overline{\xi}$ the threshold $p^*$ at which sales occur does not depend on the transaction cost on purchases.
\end{cor}

Note that for $\xi < \overline{\xi}$ the threshold $p^*$ at which sales occur does depend on the transaction cost on purchases.
This is the generic result: in {\em Case II} if the problem is well-posed then the locations of both boundaries of the no-transaction wedge depend on the values of both transaction costs.

At first sight Corollary~\ref{cor:surp} may appear surprising. At a mathematical level, the result is a consequence of the fact that the relevant solution $n_r$ passes through the singular point $(1,m(1))$ and on doing so `forgets' its starting point $(r, m(r))$. Hence $q^*$ does not depend on $\xi$ for $\xi \geq \overline{\xi}$. Since $p^* = \frac{q^*}{(1-\gamma) + \gamma q^*}$ we find that $p^*$ does not depend on $\lambda^*$.
The financial explanation of this result is fairly simple also.
If the no-transaction wedge intersects both quadrants in the upper half plane, then it includes the half-line $(x=0, y \theta>0)$. If ever $x=0$, then the agent finances consumption first by borrowing and then when borrowing levels become too great, from sales of the risky asset. But, once cash-wealth is non-positive,
the agent will trade in such a way that cash wealth is never positive at any future moment. Thus, once $x=0$ the agent will never again purchase units of risky asset, and the transaction cost on purchases becomes irrelevant. Hence the location of the sell threshold does not depend on $\xi$. The same arguments show that the value function in the second quadrant does not depend on $\xi$ (for $\xi \geq \overline{\xi}$), although it continues to depend on $\xi$ in the first quadrant.

Note that if $\xi = \overline{\xi}$ then there is no solution of the free-boundary problem (FBP). Instead the solution we want has an endpoint at $q_* =1$ at which $n'(q_*) \neq 0$, and is a solution of the initial-value problem (IVP)
\begin{quote} find $n$, $q^*$ such that $n$ is a nonnegative solution of \eqref{eqn:node} in $[1,q^*]$ with boundary conditions $n(1) = m(1)$ and $n(q^*) = m(q^*)$.
\end{quote}

\subsection{Case 1AbIi: $R<1$, $\max \{ \delta^2 R, \delta\sqrt{\frac{2R}{1-R}} \} <\epsilon< \frac{1}{1-R} + \frac{\delta^2 R}{2}$. }
This case combines the novel features of {\em Case 1AbIIi} and {\em Case 1AbIii}. See Figure~\ref{fig:Genericnq6}. The value of $m$ at the turning point is negative and so for very small transaction costs $\xi \leq \underline{\xi}$ the problem is ill-posed. For moderate transaction costs the solution is such that the no-transaction wedge lies in the second quadrant. For larger transaction costs the axis $x=0$ lies inside the no-transaction wedge, and then the ratio $q^*$ defining the sales boundary the of no-transaction wedge in the second quadrant does not depend on $\xi$.

The critical threshold $\underline{\xi}$ for well-posedness is given by $\underline{\xi} = e^{\underline{\Lambda}}-1$ where
$\underline{\Lambda}$  is given in Corollary~\ref{cor:underLam}. The critical threshold $\overline{\xi}$ above which the no-transaction wedge includes the half-line $x=0$ is not given by an algebraic expression, but instead by $\overline{\xi}=e^{\Lambda(1)}-1$ where $\Lambda(1)$ is as given in \eqref{eq:Lambda1}.

Note that for $r>1$, $n_r(q) < n_1(q)$ over the domain where both are defined, and hence $\Lambda(q_-) < \Lambda(1)$. Thus $\underline{\xi} < \overline{\xi}$.

\begin{figure}[!htbp]
	\captionsetup[subfigure]{width=0.4\textwidth}
	\centering
\subcaptionbox{Typical solutions $n_r$ together with $m$ and $\ell$.
	 \label{case6n}}{\includegraphics[scale =0.35] {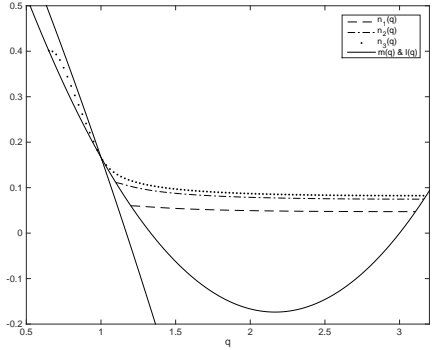}}	
\subcaptionbox{$q_* = \Sigma^{-1}(\xi)$ and $q^* = \zeta (q_*)$
     \label{case6q}}{\includegraphics[scale =0.35] {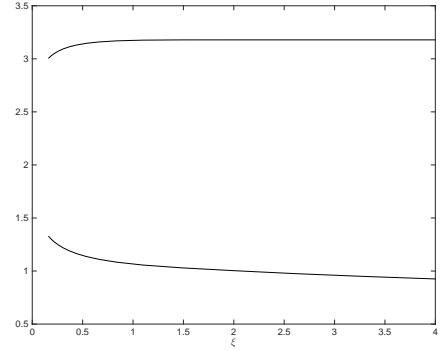}}
\caption{ {\small{
{\it Case 1AbIi. Parameter values are $\epsilon = 13/4$, $\delta = 3/2$ and $R=2/3$.
For small $\xi$ there is no solution. For large $\xi$ the candidate solution we want passes through the singular point $(1,m(1))$, and the value of $q^*$ is independent of $\xi$}.
}}}
\label{fig:Genericnq6}
\end{figure}

\subsection{Case 2AI: $R>1$, $\delta^2 R <\epsilon$.}
In this case $m'(1)>0$, and as in {\em Case 1AbIii} for any $0<r<1$ we have $\zeta(r) > 1$. See Figure~\ref{fig:Genericnq7}. Since $R>1$ the problem cannot be ill-posed. For small transaction costs (where small is $\xi < \overline{\xi}:= e^{\Lambda(1)}-1$) we find that the no transaction wedge lies strictly inside the second quadrant, and the agent always holds negative cash wealth. For $\xi \geq \overline{\xi} = e^{\Lambda(1)}-1$ the no-transaction region contains the half-line $(x=0, y \theta>0)$ and the value of $q^*$ does not depend on $\xi$.

\begin{figure}[!htbp]
	\captionsetup[subfigure]{width=0.4\textwidth}
	\centering
\subcaptionbox{Typical solutions $n_r$ together with $m$ and $\ell$.
	 \label{case7n}}{\includegraphics[scale =0.35] {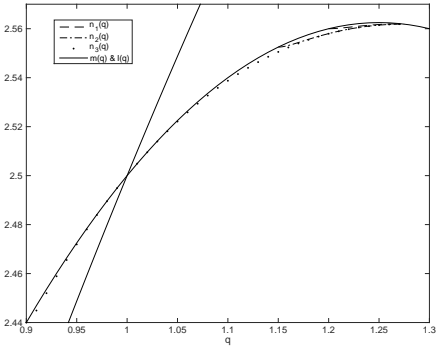}}	
\subcaptionbox{$q_* = \Sigma^{-1}(\xi)$ and $q^* = \zeta (q_*)$
     \label{case7q}}{\includegraphics[scale =0.35] {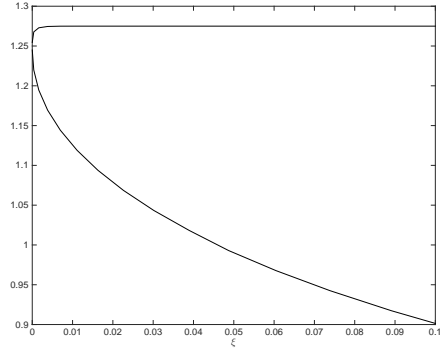}}
\caption{ {\small{
{\it Case 2AI. Parameter values are $\epsilon = 5/2$, $\delta = 1$ and $R=2$.
As in the previous case, for large $\xi$ the candidate solution we want passes through the singular point $(1,m(1))$, and the value of $q^*$ is independent of $\xi$}.
}}}
\label{fig:Genericnq7}
\end{figure}

\subsection{Case 1Bii: $R<1$, $- \delta \sqrt{\frac{2R}{1-R}} < \epsilon < 0$.}
Now $m$ is a quadratic which is increasing at zero and everywhere positive. See Figure~\ref{fig:Genericnq8}. We have $q_M<0$ and so the Merton line lies in the fourth quadrant. We are interested in the behaviour of $m$ for $q \leq 0$. For $q < 0$ we have $\ell(q) < m(q)$ and on $n>m(q)$, $O(q,n)>0$. For $q_M<r<0$ let $n_r = (n_r(q))$ solve \eqref{eqn:node} on the domain $q \leq r$ and let $\zeta(r) = \sup \{ q \leq r : n_r(q) < m(q)$. It is convenient to rewrite the definition of $\Lambda$ as
\[ \Lambda(r) = \int_{\zeta(r)}^r dq \frac{1}{|q|(1-q)}\frac{n_r(q) - m(q)}{n_r(q) - \ell(q)} \]
As in {\em Case 1AbIIii}, $\Lambda$ is increasing with $\Lambda(q_M)=0$ and $\Lambda(0)=\infty$. Setting $\Sigma(r) = e^{\Lambda(r)}-1$, $\Sigma$ is continuous and strictly increasing on $[q_m,0)$ with well-defined inverse. Set $q^* = \Sigma^{-1}(\xi)$ and $q_* = \zeta(q^*)$. The the solution to the free boundary problem is given by $n_{q^*}$ on $[q_*,q^*]$.

\begin{figure}[!htbp]
	\captionsetup[subfigure]{width=0.4\textwidth}
	\centering
\subcaptionbox{Typical solutions $n_r$ together with $m$ and $\ell$.
	 \label{case8n}}{\includegraphics[scale =0.35] {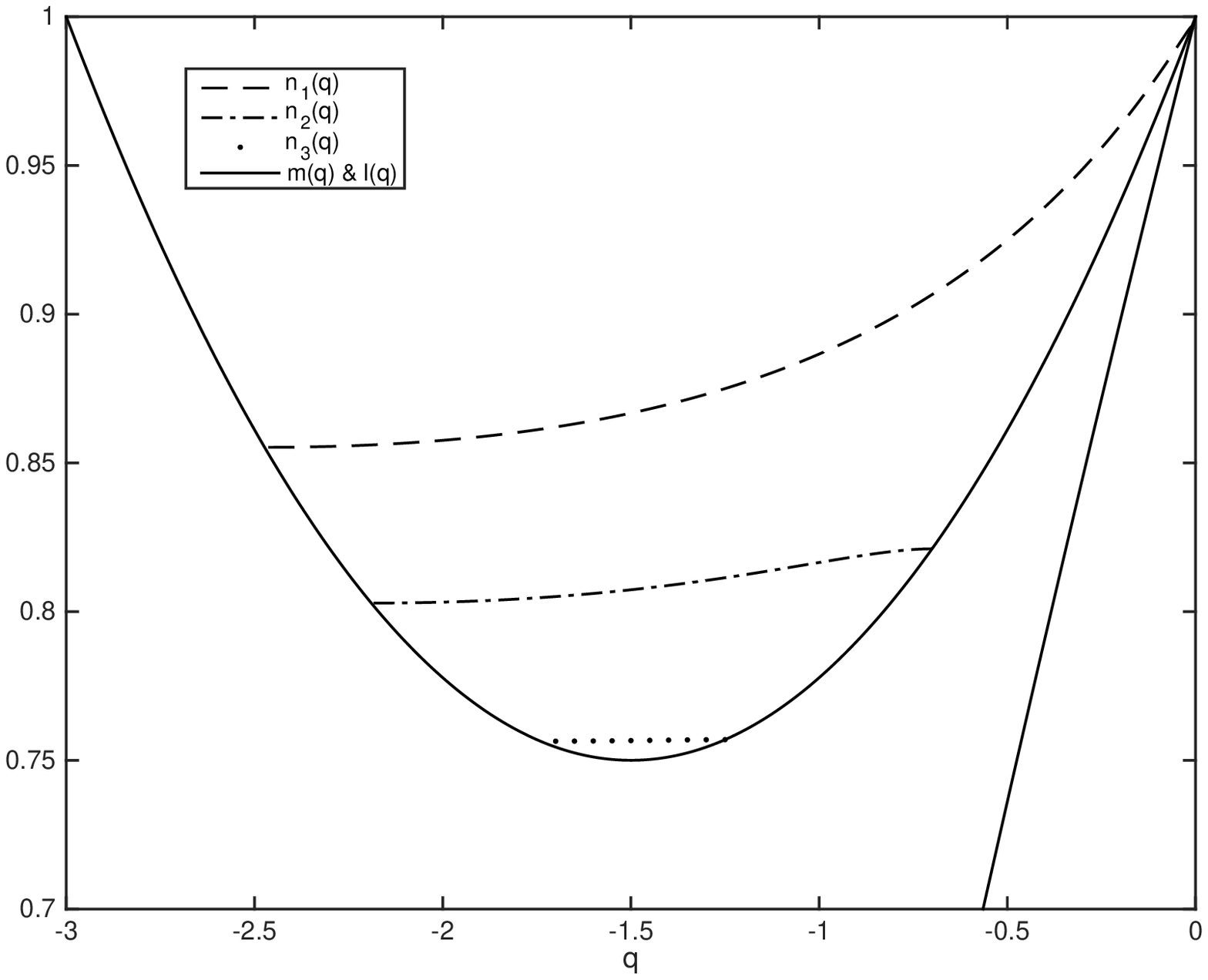}}	
\subcaptionbox{$q^* = \Sigma^{-1}(\xi)$ and $q_* = \zeta (q^*)$. For $\epsilon < 0$ ({\em Case B}) we define solutions $n_r(q)$ for $q \leq r$.
     \label{case8q}}{\includegraphics[scale =0.35] {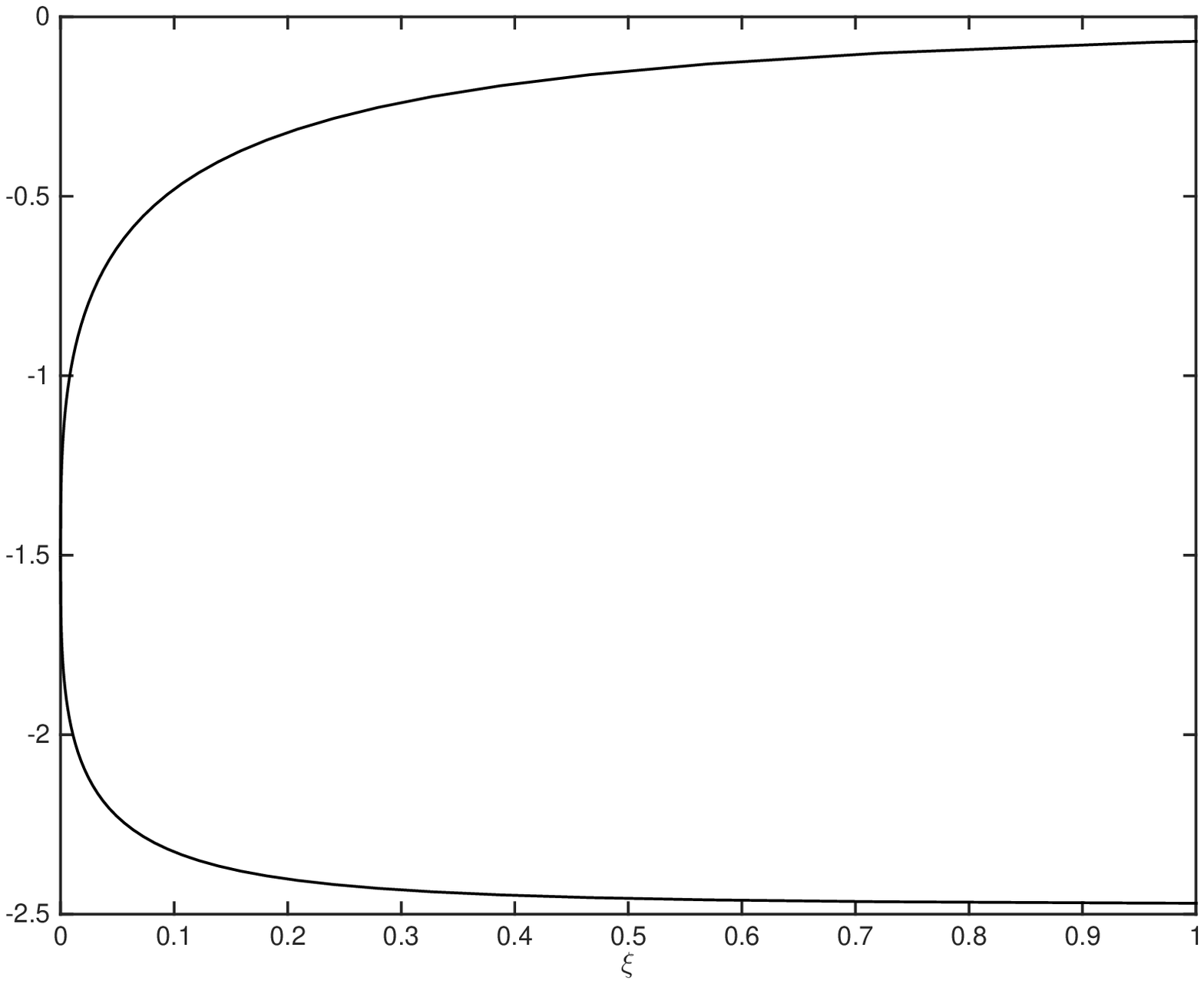}}
\caption{ {\small{
{\it Case 1Bii. Parameter values are $\epsilon = -1$, $\delta = 1$ and $R=2/3$.
Now $\epsilon <0$ and we are interested in $m$, $\ell$ and $n$ on $q<0$}.
}}}
\label{fig:Genericnq8}
\end{figure}

\subsection{Case 1Bi: $R<1$, $\epsilon < - \delta \sqrt{\frac{2R}{1-R}}$.}
This case combines the features of {\em Cases 1AaIIi} and {\em 1Bii}. The no-transaction wedge lies in the fourth quadrant, but for sufficiently small transaction costs the problem is ill-posed. See Figure~\ref{fig:Genericnq9}. The round-trip transaction cost threshold between ill-posed and well-posed problems is given by $\underline{\xi} = e^{\underline{\Lambda}}-1$ where $\underline{\Lambda}$ is as given in Corollary~\ref{cor:underLam}.

\begin{figure}[!htbp]
	\captionsetup[subfigure]{width=0.4\textwidth}
	\centering
\subcaptionbox{Typical solutions $n_r$ together with $m$ and $\ell$.
	 \label{case9n}}{\includegraphics[scale =0.35] {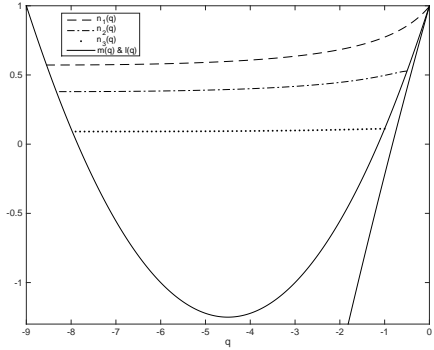}}	
\subcaptionbox{$q^* = \Sigma^{-1}(\xi)$ and $q_* = \zeta (q^*)$
     \label{case9q}}{\includegraphics[scale =0.35] {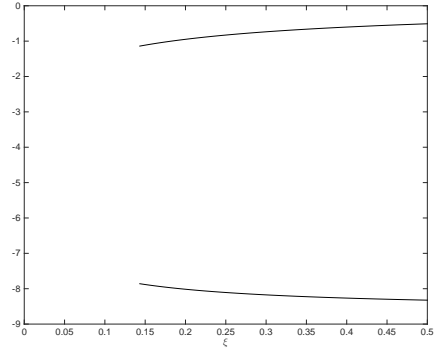}}
\caption{ {\small{
{\it Case 1Bi. Parameter values are $\epsilon = -3$, $\delta = 1$ and $R=2/3$.
For small transaction costs the problem is ill-posed}.
}}}
\label{fig:Genericnq9}
\end{figure}

\subsection{Case 2B: $R>1$, $\epsilon < 0$.}
In this case $\ell(q)>m(q)$ on $q < 0$ and $O(q,n)<0$ on $q<0$ and $n < m(q) < \ell(q)$. The solutions we want are decreasing on $q_* <  q < q^*$, and a solution with $q_* < q_M < q^* < 0$  exists for each possible round-trip transaction cost. See Figure~\ref{fig:Genericnq10}.

\begin{figure}[!htbp]
	\captionsetup[subfigure]{width=0.4\textwidth}
	\centering
\subcaptionbox{Typical solutions $n_r$ together with $m$ and $\ell$.
	 \label{case10n}}{\includegraphics[scale =0.35] {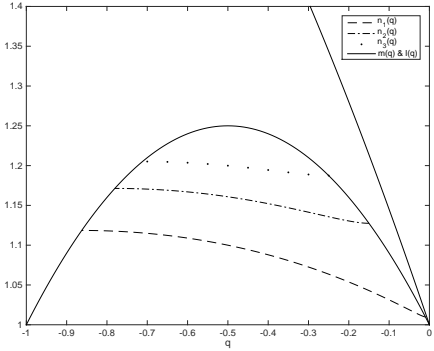}}	
\subcaptionbox{$q^* = \Sigma^{-1}(\xi)$ and $q_* = \zeta (q^*)$
     \label{case10q}}{\includegraphics[scale =0.35] {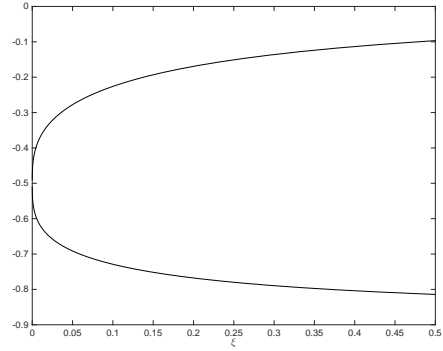}}
\caption{ {\small{
{\it Case 2B. Parameter values are $\epsilon = -1$, $\delta = 1$ and $R=2$.
Since $R>1$ we have $n<m<\ell$ on the region of interest}.}}}
\label{fig:Genericnq10}
\end{figure}

\subsection{Boundary cases}
If $m'(0)=0$ or equivalently $\epsilon = 0$ then $q_*=q^*=q_M=0$. It is optimal to sell any initial endowment in the risky asset immediately. Thereafter no further trading is required. Note that in this case the solution to the Merton problem (with no transaction costs) is also to have zero investment in the risky asset since holding the asset brings risk but no return. For this reason it better to sell now rather than later.

If $R<1$ and $m(1)=0$ or equivalently $\epsilon = \frac{1}{1-R} + \frac{\delta^2 R}{2}$ then the problem is ill-posed.

If $m'(1)=0$ (and $m(1)>0$) or equivalently $\epsilon = \delta^2 R$ (and, if $R<1$, $\epsilon < \frac{2}{1-R}$) then $q_* < q^* = 1$ for any value of the transaction cost. The optimal solution to the Merton problem is to invest exclusively in the risky asset and to keep zero cash holdings. Consumption is financed from sales of the risky asset. For the transaction cost problem, if ever cash wealth hits zero, then the investor keeps cash wealth at zero (and finances consumption from sales of the risky asset). But if the agent has a small positive cash wealth, then before selling any risky asset, he first finances consumption from cash wealth. In fact, if there are adverse movements in the price of the risky asset, the agent may also purchase units of risky asset.

If $m_M=0$ (and $m'(1) \neq 0$) or equivalently $\epsilon = \delta \sqrt{\frac{2 R}{1-R} }$ (and $\delta^2 \neq \frac{2}{R(1-R)}$) the problem is ill-posed for zero-transaction costs, but well posed for any positive level of round-trip transaction cost.


\section{Verification Lemmas}
\label{sec:verification}

The main theorem of this paper is the following

\begin{thm}
Recall the formula for $\underline{\Lambda}$ in \eqref{eq:undLamdef} and set $\underline{\xi}=e^{\underline{\Lambda}}-1$.
\begin{enumerate}
\item Suppose either
\begin{enumerate}
\item $R>1$ or
\item $R<1$ and $|\epsilon| < \delta \sqrt{\frac{2R}{1-R}}$ or
\item $R<1$, $\delta \sqrt{\frac{2R}{1-R}} < \epsilon < \frac{1}{1-R} + \frac{\delta^2 R}{2}$ and $\xi > \underline{\xi}$ or
\item $R<1$, $\epsilon < - \delta \sqrt{\frac{2R}{1-R}}$ and $\xi > \underline{\xi}$.
\end{enumerate}
Then the problem is well-posed.
\item
Suppose either
\begin{enumerate}
\item $R<1$ and $\epsilon > \frac{1}{1-R} + \frac{\delta^2 R}{2}$ or
\item $R<1$, $\delta \sqrt{\frac{2R}{1-R}} < \epsilon < \frac{1}{1-R} + \frac{\delta^2 R}{2}$ and $\xi \leq \underline{\xi}$ or
\item $R<1$, $\epsilon < - \delta \sqrt{\frac{2R}{1-R}}$ and $\xi \leq \underline{\xi}$.
\end{enumerate}
Then the problem is ill-posed.
\end{enumerate}
\end{thm}

\begin{proof}
Since the majority of this result is contained in Choi et al~\cite{ChoiSirbuZitkovic:13} we only provide a sketch of the proof. Proofs of well-posedness for subsets of the parameter combinations can also be found in Davis and Norman~\cite{DavisNorman:90} and Herczegh and Prokaj~\cite{HerczeghProkaj:15}.
The main innovations compared with \cite{ChoiSirbuZitkovic:13} are that we cover the case $\epsilon \leq 0$ and we give an explicit formula for $\underline{\xi}$. The other substantial difference is that we take a classical approach via the value function and the Hamilton-Jacobi-Bellman equation, whereas Choi et al construct a solution via the dual problem and the shadow price.

Our contention is that whilst the two approaches are equivalent, ultimately our analysis is simpler, in the sense that the solutions are characterised by the behaviour of a simple quadratic function.

{\em The well-posed case:}
For the parameter combinations listed as leading to a well-posed problem (excluding for a moment the case where $\xi = \overline{\xi}$) we can construct a positive, $C^1$-solution $n$ to the free-boundary problem and thence a function $G$ and a candidate value function $V^C$ given by $V^C(X_t,Y_t,\Theta_t,t) = e^{-\beta t} \frac{(x+y \theta)^{1-R}}{1-R} G(\frac{y \theta}{x+y\theta} )$. It remains to prove that this candidate value function is the value function $V$ of the optimal consumption/investment problem.

Note that since $n$ is $C^1$ we have that the candidate value function $V^C$ is $C^2$ on the solvency region. Hence we can apply It\^{o}'s formula.
Set $M_t = V^C(x,y,\theta,t) + \int_0^t e^{-\beta s} \frac{{C_s}^{1-R}}{1-R} ds$. Then under any admissible strategy $dM_t \leq \sigma Y_t V^C_y(X_t,Y_t,\Theta_t,t) dB_t$ and $M_t \leq M_0 +\hat{M}_t=V^C(x,y,\theta,0)+\hat{M}_t$ where $\hat{M}_t = \int_0^t \sigma Y_t V_y dB_t$.

Suppose $R<1$. Then $\hat{M}$ is a local martingale null at 0. Also $M \geq 0$ and so the local martingale $\hat{M}$ is bounded below by $-M_0$ and hence a supermartingale. Then $\E[{M}_t] \leq M_0$ and for any admissible strategy
\[ \E \left[ \int_0^t e^{-\beta s} \frac{C_s^{1-R}}{1-R} ds \right] \leq \E \left[ \int_0^t e^{-\beta s} \frac{C_s^{1-R}}{1-R} ds \right] + \E \left[ V^C(X_t,Y_t,\Theta_t,t)\right]= \E[M_t] \leq M_0 \]
and by monotone convergence, for any admissible strategy
\[ \E \left[ \int_0^\infty \frac{C_s^{1-R}}{1-R} ds \right] = \lim \E \left[ \int_0^t \frac{C_s^{1-R}}{1-R} ds \right] \leq M_0 = V^C(x,y,\theta,0). \]
Hence $V \leq V^C$.

To show the converse, we need to exhibit an admissible strategy for which
$\E \left[ \int_0^\infty e^{-\beta s} \frac{C_s^{1-R}}{1-R} ds \right] = V^C(x,y,\theta,0)$. Then $V \geq V^C$ and we are done.
Let $(C^*,\Theta^*)$ be the candidate optimal strategy, with associated wealth process $X^*$. Then,
\[ C^*_t = \left[ G(P^*_t) - \frac{P^*_t G'(P^*_t)}{1-R} \right]^{-1/R} \frac{\beta}{R} (X^*_t + \Theta^* Y_t) \]
and $\Theta_t^*$ is the singular-control, local time strategy which involves selling/purchasing just enough risky asset to keep $P_t$ in the interval $[p_*,p^*]$.
Define $M^*$ via
$M^*_t = \int_0^t \sigma Y_s V^C_y(X^*_s, Y_s, \Theta^*_s,s) dB_s$. Then since $V^C$ includes a factor which decays exponentially over time, it is possible to show that over any horizon $T$, $(M^*_t)_{0 \leq t \leq T}$ is a martingale and $V^C(X_T^*,Y_T,\Theta^*_T,T) \rightarrow 0$ almost surely and in $L^1$ (see \cite{DavisNorman:90,HobsonZhu:14,Zhu:15} for this result). Then, from the martingale property of $M^*$,
\[ V^C(x,y,\theta,0) = \E[V^C(X^*_T,Y_T,\theta_T, T)] + \E \left[ \int_0^T e^{-\beta s} \frac{(C^*_s)^{1-R}}{1-R} ds \right] \]
and letting $T \uparrow \infty$ we conclude
 \[ V^C(x,y,\theta,0) =  \E \left[ \int_0^\infty \frac{(C^*_s)^{1-R}}{1-R} ds \right] \leq V(x,y,\theta,0) \]
 as required.

If $R>1$ the argument that the local martingale $\hat{M}$ is a super-martingale fails, since it is not bounded below. However Davis and Norman~\cite{DavisNorman:90}, see also \cite{HobsonZhu:14,Zhu:15}, give an ingenious argument based on taking limits for a family of perturbed utility functions to show that the result holds for $R>1$ also.

If $\xi = \overline{\xi}$ then the same ideas work, except that we construct the value function via $n$ which solves the modified problem whereby
$n = n_1(q)$, $n(q^*)=m(q^*)$. Then $n'(1)=m'(1) \neq 0$. Note there is still smooth fit at $q_*=1$, but no second order smooth fit. Hence $G''$ is discontinuous at $p=1$ but we can still apply It\^{o}'s formula to \eqref{eq:Mdef}. In this case the $y \theta$ axis is a boundary to the no-transaction region. Once the agent reaches a leveraged position and cash wealth is negative then cash wealth remains negative for evermore.

{\em The ill-posed case:}
In this case it is sufficient to exhibit a strategy which yields infinite expected utility. See Choi et al~\cite{ChoiSirbuZitkovic:13} for details.
\end{proof}

\section{The dependence of the no-transaction wedge on parameters}
\label{sec:statics}

\subsection{Dependence on transaction costs}

For most of this paper we have argued that the transaction costs $\lambda$ on purchases and $\gamma$ on sales only enter the problem through the round-trip transaction cost $\xi = \frac{\lambda + \gamma}{1-\gamma}$. Whilst that is true for the construction of $n$ (and the locations of the free-boundaries $q_*$ and $q^*$ from which the solution is built), the boundaries of the no-transaction wedge do depend on the individual transaction costs and we have  $p_* =p_*(\lambda,\gamma)$
and $p^*=p^*(\lambda,\gamma)$ where
\begin{equation}
 \label{eq:p*p*}
 p_*(\lambda,\gamma) = \frac{q_*(\xi)}{1 + \lambda - \lambda q_*(\xi)}; \hspace{20mm} p^*(\lambda,\gamma) = \frac{q^*(\xi)}{1 - \gamma + \gamma q^*(\xi)},
\end{equation}
and $\xi = \frac{\lambda + \gamma}{1-\gamma}$.

\begin{thm}
Suppose all parameters are fixed, except for the transaction costs $\lambda$ and $\gamma$. Suppose the problem is well-posed. Then
\begin{itemize}
\item[(a)]
\begin{itemize}
\item[(i)] $q^*$ is non-decreasing in $\xi$, and $q_*$ is increasing in $\xi$.
\item[(ii)] If $\epsilon > \delta^2 R$ and $\xi \geq \overline{\xi}$ then $q_* \leq 1$ and $q^*$ does not depend on $\xi$.
\end{itemize}
\item[(b)]
\begin{itemize}
\item[(i)] If $0 < \epsilon < \delta^2 R$ then the purchase boundary $p_*$ and sale boundary $p^*$ are increasing in $\lambda$ and increasing in $\gamma$ and the Merton line lies within the no-transaction wedge. We have $0 < p_* < q_M = \frac{\epsilon}{\delta^2 R} < p^* < 1$.
\item[(ii)]
Suppose $\epsilon > \delta^2 R$ or $\epsilon < 0$. Then $p_*$ and $p^*$ need not be monotonic in the individual transaction costs, and the Merton line need not lie within the no-transaction wedge.
If $\epsilon > \delta^2 R$ then $p^* > 1$ and the agent will (at least sometimes) take a leveraged position. If $\epsilon < 0$ then $p_* < p^* < 0$ and the agent will take a short position.
\end{itemize}
\end{itemize}
\end{thm}

\begin{proof}
(a)(i) Recall the definitions $\zeta(r) = \inf \{ u>r : n_r(u) > m(u) \}$ and
\[ \Lambda(r) = \int_r^{\zeta(r)} \frac{1}{q(1-q)} \frac{n_r(q)-m(q)}{\ell(q) - n(q)} dq. \]
From the non-crossing property of the solutions $n_r$ it follows that $\Lambda$ is decreasing in $r$ and that $q^*$ is increasing in $\xi$, and $q^*$ is decreasing in $\xi$. The implicit function theorem gives that away from $\xi=0$ and $\xi = \overline{\xi}$, $q_*$ and $q^*$ are differentiable.

(a)(ii) This follows from Lemma~\ref{lem:singular}.

(b) We have
\[ \frac{d p^*}{d \lambda} = \frac{\partial \xi}{\partial \lambda} \frac{\partial q^*}{\partial \xi} \frac{d p^*}{d q^*} =
 \frac{1}{(1-\gamma)} \frac{\partial q^*}{\partial \xi} \frac{1-\gamma}{(1-\gamma+\gamma q^*)^2} > 0\]
and
\[ \frac{d p^*}{d \gamma} = \frac{q^*(1-q^*)}{( 1- \gamma(1-q^*))^2} + \frac{\partial \xi}{\partial \gamma} \frac{\partial q^*}{\partial \xi} \frac{d p^*}{d q^*}
= \frac{q^*(1-q^*)}{( 1- \gamma(1-q^*))^2} +\frac{1 + \lambda}{(1-\gamma)^2} \frac{\partial q^*}{\partial \xi} \frac{1-\gamma}{(1-\gamma+\gamma q^*)^2} \]
If $0<\epsilon<\delta^2R$ then $0<q^*<1$ and the sign of both terms is positive, but if $q^*\notin [0,1]$ then either term may dominate.

Similarly,
\[ \frac{d p_*}{d \gamma} = \frac{\partial \xi}{\partial \gamma} \frac{\partial q_*}{\partial \xi} \frac{d p_*}{d q_*}
= \frac{1+\lambda}{(1-\gamma)^2} \frac{\partial q_*}{\partial \xi} \frac{1+\lambda}{(1+\lambda-\lambda q_*)^2} < 0\]
and
\[ \frac{d p_*}{d \lambda} =
\frac{-q_*(1-q_*)}{(1 +\lambda(1-q_*))^2} +
\frac{\partial \xi}{\partial \lambda} \frac{\partial q_*}{\partial \xi} \frac{d p_*}{d q_*} = \frac{-q_*(1-q_*)}{(1 +\lambda- \lambda q_*))^2}
+ \frac{1}{(1-\gamma)} \frac{\partial q_*}{\partial \xi} \frac{1+\lambda}{(1 + \lambda - \lambda q_*)^2} \]
If $0<\epsilon<\delta^2R$ then $0<q^*<1$ and the sign of both terms is negative, but if $q_* \notin[0,1]$ then either term may dominate.

Note that solvency requires that $p^* < \frac{1}{\gamma}$. So, when $\epsilon > \frac{\delta^2 R}{\gamma}$ we have $1 < p^* < \frac{1}{\gamma} < q_M$ and the Merton line lies outside the
no transaction wedge.
\end{proof}

Of interest is the location of the no-transaction wedge and the relationship between the Merton line and the no-transaction wedge. The key advantage we have over the previous literature (\cite{DavisNorman:90,ShreveSoner:94}) is that we have decoupled the expressions for the locations of the boundaries of the no-transaction wedge into two parts: we have $p_*$ and $p^*$ given by \eqref{eq:p*p*} where $q_* < q_M < q^*$.

Davis and Norman~\cite{DavisNorman:90} argue that if $0< \epsilon < \delta^2 R \wedge \delta \sqrt{\frac{2R}{1-R}}$ (and a further technical condition, Condition B holds) then the no-transaction wedge lies in the first quadrant and contains the Merton line. We saw this in the final comments of Section~\ref{sec:sc}. They also conjecture \cite[p704]{DavisNorman:90} that if the problem is well-posed and $\epsilon > \delta^2R$ then the no-transaction wedge lies in the second quadrant. As we have seen, if transaction costs are sufficiently large, this need not be the case.

Shreve and Soner~\cite{ShreveSoner:94} give bounds on $p_*$ and $p^*$. They state in (11.4), (11.5) and (11.6) of \cite{ShreveSoner:94} that
\begin{equation}
 p^*  <  \frac{\epsilon}{\frac{1}{2} (1-\gamma) \delta^2 R + \gamma \epsilon} ;
 \label{eq:pupup}
\end{equation}
if $0< \epsilon < \delta \sqrt{ \frac{2R}{1-R} }$
\begin{equation}
p^* > \frac{\epsilon}{(1-\gamma) \delta^2 R + \gamma \epsilon} ;
\label{eq:pupdown}
\end{equation}
and if $0< \epsilon < \delta \sqrt{ \frac{2R}{1-R} }$ and $\epsilon < \delta^2R \frac{1+\lambda}{\lambda}$
\begin{equation}
p_*  <  \frac{\epsilon}{(1+\lambda) \delta^2 R - \lambda \epsilon}.
\label{eq:pdownup}
\end{equation}

The bounds \eqref{eq:pupup}, \eqref{eq:pupdown} and \eqref{eq:pdownup} can be seen to follow from our results, sometimes under weaker assumptions.

If $0< q_M < 1$ (equivalently $0 < \epsilon < \delta^2 R$) and the problem is well-posed then since $m$ is a quadratic and $n$ is monotone we must have $(1-R)m(q^*) = (1-R)n(q^*) < (1-R)n(q_*) = (1-R)m(q_*)$ and so $q^* - q_M < q_M - q_* < q_M$. We conclude that $q^* < \min \{ 2 q_M, 1 \}$.
Then, since $q_* < q_M < q^*$,
\begin{eqnarray}
p^* &= &\frac{q^*}{(1-\gamma) + \gamma q^*} < \frac{2q_M}{(1-\gamma) + \gamma 2q_M} = \frac{\epsilon}{\frac{1}{2}(1-\gamma)\delta^2 R + \gamma \epsilon}; \label{eq:mypupup} \\
p^* &= &\frac{q^*}{(1-\gamma) + \gamma q^*} > \frac{q_M}{(1-\gamma) + \gamma q_M} = \frac{\epsilon}{(1-\gamma)\delta^2 R + \gamma \epsilon}; \label{eq:mypupdown} \\
0 < p_*  &= &\frac{q_*}{(1+\lambda) - \lambda q_*} < \frac{q_M}{(1+\lambda) - \lambda q_M} = \frac{\epsilon}{(1+\lambda)\delta^2 R - \lambda \epsilon}. \label{eq:mypdownup}
\end{eqnarray}
Note that from $q^*<1$ we also have the bound $p^*<1$, and the no-transaction wedge lies in the first quadrant.

If $q_M>1$ (equivalently $\epsilon > \delta^2 R$) and the problem is well-posed then since $(1-R)m(q^*) = (1-R)n(q^*) < (1-R)n(1) = (1-R)m(1)$ we have $q^* - q_M < q_M - \max \{q_*,1 \} \leq q_M-1$.
Then $1 < q_M < q^* < 2 q_M - 1$ and \eqref{eq:mypupup} can be refined to
\[
p^* < \frac{2q_M-1}{(1-\gamma) + \gamma (2q_M-1)} = \frac{2\epsilon - \delta^2 R}{(1-2\gamma)\delta^2 R + 2 \gamma \epsilon}.
\]
\eqref{eq:mypupdown} and \eqref{eq:mypdownup} hold as before, \eqref{eq:mypdownup} provided $\epsilon < \delta^2R \frac{1+\lambda}{\lambda}$.

Shreve and Soner~\cite[p675]{ShreveSoner:94} also conjecture that if $q_M>1$ then $p^* < q_M$
and the Merton line lies outside the no-transaction wedge. If $q_M>1$ then we have $q^* < 2q_M -1$ and $p^* <\frac{2 q_M - 1}{ (1-\gamma) + \gamma (2 q_M -1 )}$. Then if $\frac{1}{2 q_M} < \gamma <1$ so that transaction costs on sales are large we have $p^* <\frac{2 q_M - 1}{ (1-\gamma) + \gamma (2 q_M -1 )} < q_M$ and the Shreve-Soner conjecture is true. However, if transaction costs on sales are small we may find $p_*< q_M < p^*$, and the Merton line lies inside the no-transaction wedge. 
In particular, if $\gamma=0$ then $p^*=q^*$ and $p^*>q_M$.

\subsection{Dependence on drift}
\begin{thm}
Suppose all parameters except the drift are constant and that the problem is well posed. Then both the purchase and sale boundaries of the no-transaction wedge are increasing in the drift in the underlying asset.
\end{thm}

\begin{proof} We want to show that both $p_*$ and $p^*$ are increasing in $\mu$, which is equivalent to $q_*$ and $q^*$ increasing in $\epsilon$. We consider the case $\epsilon >0$; similar arguments work for $\epsilon < 0$.

Fix $\hat{\epsilon} > \tilde{\epsilon}$ and let $\hat{n}_r$ and $\tilde{n}_r$ denote the solutions of
$n' = O(q, \hat{m}, n)$ and $n'=O(q,\tilde{m},n)$ subject to $n_r(r)=0$ where
\[ O(q,m,n) = - \frac{1-R}{R} \frac{n}{1-q} \frac{n - m(q)}{\left( m(q) + \frac{\delta^2}{2}(1-R) q(1-q) - n \right)} . \]
Here $\hat{m}(q)$ (respectively $\tilde{m}$) is the quadratic $\hat{m}(q)=1-\hat{\epsilon}(1-R)q + \frac{\delta^2}{2}R(1-R)q^2$ (respectively $\tilde{m}(q)=1-\tilde{\epsilon}(1-R)q + \frac{\delta^2}{2}R(1-R)q^2$).
In general let the $\hat{\cdot}$ and $\tilde{\cdot}$ symbols denote solutions defined relative to $\hat{\epsilon}$
and $\tilde{\epsilon}$. Let $m_0(q) = 1 + \frac{1}{2} \delta^2R(1-R)q^2$.

Let $\hat{a}(q) = \hat{a}_r(q) = \hat{n}_r(q)- \hat{m}(q)$. Then
\begin{eqnarray*}
\hat{a}'(q) & = & O(q, m_0(q) - \hat{\epsilon}(1-R)q, m_0(q) - \hat{\epsilon}(1-R)q + \hat{a}(q)) + \hat{\epsilon}(1-R) - \delta^2 R(1-R) q \\
& = & O(q, m_0(q), m_0(q) + \hat{a}) - \delta^2 R(1-R) q +
  \hat{\epsilon}(1-R) \left[ \frac{1-R}{R} \frac{q}{(1-q)} \frac{\hat{a}}{(\frac{1}{2} \delta^2 (1-R)q(1-q) - \hat{a})} + 1 \right] \\
  & =: & \hat{O}(q, \hat{a}).
\end{eqnarray*}
For $R<1$ and $q<1$ we have $0<a<\ell(q)-m(q)$ and $\hat{O}(q,a) > \tilde{O}(q, a)$,
and we conclude that away from $q=r$, $\hat{a}_r$ and $\tilde{a}_r$ cannot cross. Consideration of the cases for $R>1$, and/or $q>1$ leads to the a similar conclusion.

Suppose first that $\hat{\epsilon} < \delta^2 R$ so that we may restrict attention to $r<q<1$.
Fix $r$. Then $\hat{a}_r(q) > \tilde{a}_r(q)$ at least until $\hat{\zeta}(r) \wedge \tilde{\zeta}(r)$
and then it follows both that $\hat{\zeta}(r) > \tilde{\zeta}(r)$ and $\hat{\Lambda}(r) > \tilde{\Lambda}(r)$, where we make use of $a=n-m$ and the representation
\[ \Lambda(r) = \int_r^{\zeta(r)} \frac{dq}{q(1-q)} \frac{a(q)}{\frac{\delta^2}{2}R(1-R)q(1-q) - a(q)} \]
which we note only depends on $\epsilon$ through $a$.
Since $\Sigma(r) = e^{\Lambda(r)}-1$ is decreasing in $r$ we conclude that $\hat{q}_* = \hat{\Sigma}^{-1}(\xi) > \tilde{\Sigma}^{-1}(\xi) = \tilde{q}_*$ and that $q_*$ is increasing in $\epsilon$. 

In order to consider the sale boundary $p^*$ it is convenient to parameterise solutions of the free boundary problem by the boundary point $q^*$ rather than $q_*$. Let $n$ solve $n'=O(q,n)$ in $q \leq s$ subject to $n_s(s)=m(s)$ and let $a_s(q) = n_s(q)-m(q)$. Then, we have $\zeta^{-1}(s) = \sup \{ u \leq s : n_s(u) < m(s) \}$ and again we get that solutions $(a_s(q))_{ \{ \zeta^{-1}(s) \leq q \leq s\} }$ are increasing in $\epsilon$; hence $\zeta^{-1}(s)$ is increasing in $\epsilon$ and $\Lambda$ is increasing in $\epsilon$. It follows that $q^*$ is also increasing in $\epsilon$. 

Now we relax the assumption that $\hat{\epsilon} < \delta^2 R$. If $\tilde{\epsilon} \leq \delta^2 R < \hat{\epsilon}$, then $\tilde{q}^* \leq 1 < \hat{q}^*$.
For $q_*$ the same proof as given above can be used.

Finally, if $\tilde{\epsilon} > \delta^2 R$ then for sufficiently small transaction costs we have $\hat{q}_* >1$, and then by the arguments as above we can conclude that
$q_*$ and $q^*$ are monotonic. The only point of delicacy is when transaction costs are larger, when we must consider the case where both $\hat{q}_*$ and $\tilde{q}_*$ lie below the singular point. Then, for $r<1$, $\hat{a}_r(1)=\tilde{a}_r(1)=0$. Nonetheless, for $r<1$ we have the inequality $\hat{a}_r \geq \tilde{a}_r$ with strict inequality on $(r,1)$, and hence $\hat{\Lambda}(r) > \tilde{\Lambda}(r)$. Note that $\hat{q}^* = \tilde{q}^*$ for large transaction costs.

\end{proof}

\section{Conclusions}
\label{sec:conc}

Our goal in this paper was to analyse the Merton problem with transaction costs via the classical approach and the primal problem. We were able to show via judicious transformations that the problem could be reduced to solving a free-boundary problem for a first-order ordinary differential equation. There is a family of solutions to this free boundary problem, and the one we want satisfies an additional integral equation.

Our first main result mirrors the main result of Choi et al~\cite{ChoiSirbuZitkovic:13}. We cover some additional cases (negative drift) but this is not our main contribution. Instead, our main contribution is to demonstrate that the different regimes (determining when the problem is well-posed for all transaction costs, when the problem is well-posed for large transaction costs only, and when the problem is ill-posed; and determining whether the no-transaction region lies in the first, second, fourth or first and second quadrants) depend on the shapes of a quadratic $m=m(q)$, its values and first derivatives at zero and one, and the value at the turning point. Choi et al~\cite{ChoiSirbuZitkovic:13} also reduce the problem to solving a first order ODE, subject to smooth fit conditions on a free-boundary, and subject to an integral condition. But in their case the points on the free-boundary lie on an ellipse (rather than a quadratic) and the phase-diagram is considerably more complicated.

Following Davis and Norman~\cite{DavisNorman:90} and Shreve and Soner~\cite{ShreveSoner:94}, our approach is via the primal problem rather than the shadow-price approach of \cite{KallsenMuhle-Karbe:10,ChoiSirbuZitkovic:13,HerczeghProkaj:15}. Thus our approach brings different insights to \cite{KallsenMuhle-Karbe:10,ChoiSirbuZitkovic:13,HerczeghProkaj:15}. At one level our results are a re-parametrisation of the results of Choi et al~\cite{ChoiSirbuZitkovic:13} although the derivation is completely different. Nonetheless, this re-parametrisation brings significant simplifications in the analysis. First, as described above we can relate the different cases to the different possible behaviours of a quadratic. Second, in the case where the problem is ill-posed with zero transaction costs, we can give an algebraic expression for the value of the transaction costs at which the problem becomes ill-posed. (Choi et al are only able to give this as an integral involving the roots of a quadratic, see Choi et al~\cite[Lemma 6.11]{ChoiSirbuZitkovic:13}.)
Third, it is immediate from our approach that the integral equation which determines which of the family of candidate solutions of the free-boundary we want has a monotonicity property. In particular, it is immediate from our approach that $\Sigma$ is strictly decreasing and has an inverse: Choi et al~\cite[Remark 6.15]{ChoiSirbuZitkovic:13} are not able to give a corresponding monotonicity argument for their equivalent function. Fourth, we can give a direct interpretation of the free-boundary points of the solution to the first order ODE in terms of the boundaries of the no-transaction wedge. This allows us to identify the solutions of the ODE which pass through the singular point as those which correspond to a no-transaction wedge which includes the half-line $(x=0, y \theta>0)$. This leads to a further insight which is not present in \cite{ChoiSirbuZitkovic:13}. Choi et al show a result which is equivalent to the fact that all the solutions of $n'=O(q,n)$ which pass through the singular point at $q=1$ are the same to the right of the singular point, and conclude that the shadow price and value function are independent of the value of transaction costs (provided the level of the transaction cost is above a certain critical value). However, they do not give a financial explanation of this result.  In contrast we can give an explanation. We observe that the value function is based on a function $n$ which passes through the singularity at $q=1$ if and only if the half-line $(x=0, y\theta>0)$ is inside the no-transaction wedge. If this line is within the no-transaction wedge then under optimal behaviour, since cash wealth only ever increases on the sale boundary, once cash wealth equals zero it is negative at all times thereafter. Hence, since the purchase boundary of the no-transaction region
is in the first quadrant, no future purchases of the risky asset will ever take place. It follows that, in the region $p>1$, the value function (and the location of the sale boundary of the no transaction region) cannot depend on the level of transaction costs on purchases.

\appendix

\section{Proofs}

\subsection{Proof of Lemma~\ref{lem:3.1}}
Our goal is to prove Lemma~\ref{lem:3.1} which is contained in the union of the following results. 
First we show the previously advertised result that for each $r \in (0,1)$, $n_r$ lies below the line joining $(0,1)$ to $(1,m(1))$ if $R<1$ and above this line if $R>1$.

\begin{lem} Suppose $\epsilon > 0$ and $m(1)>0$.
Let $c(q) = 1 + (1-R) \left(\frac{\delta^2 R}{2} - \epsilon \right)q$. Then for $0 < r < q < \zeta(r) \wedge 1$, $(1-R)n_r(q) \leq (1-R)c(q)$.
\end{lem}

\begin{proof} If $R>1$ then $m(q) > n(q) > \ell(q) > c(q)$ on $(0,1)$.

Now suppose  $R<1$.
Note that $-1 < m(1)-1$ and $-m(1)<0$ so that since $c$ is a straight line we have the inequality $-c(q) < (1-q)(m(1)-1)$ on $(0,1)$.
Then we have
\[ O(q,c(q)) = - \frac{1-R}{R} \frac{c(q)}{(1-q)} \frac{c(q)-m(q)}{\ell(q)-c(q)} = - \frac{1-R}{R} \frac{c(q)}{1-q} \frac{\frac{\delta^2 R}{2}(1-R)q(1-q)}{\frac{\delta^2}{2}(1-R)^2q(1-q)} = - \frac{c(q)}{(1-q)} < m(1)-m(0) = c'(q). \]
Hence solutions $n$ of \eqref{eqn:node} can only cross $c$ from above to below. Hence $n_r(q) \leq c(q)$.
\end{proof}

\begin{lem}
Suppose $R<1$, $m'(0)<0$, $m_M>0$, $m(1)>0$ and $m'(1)>0$. Then for $0 \leq r < q_M$, $q_M<\zeta(r) < 1$.
\end{lem}

\begin{proof} First note we have $0 = n'(\zeta(r)) < m'(\zeta(r))$ and hence $\zeta(r) > q_M$.

If $\ell'(1) \geq 0$ then $\ell$ is increasing on $(0,1)$ and $m(1)>m(0)$. Then, for $r<1$, since $n_r$ is decreasing we must have $\zeta(r)<1$.

So, suppose $\ell'(1)<0$.
Suppose for a contradiction that $\zeta(r)\geq 1$. Then $n_r$ is decreasing on $(r,1]$ and $n_r(q) \in (m(1),\ell(q))$. Then $O(q,n(q)) < -\frac{1-R}{R} \frac{m(1)}{1-q} \frac{m(1)- m(q)}{\ell(q) - \ell(1)}$ and as $q \uparrow 1$, $\lim_{q}(1-q))O(q,n(q)) <  \frac{1-R}{R} m(1) \frac{m'(1)}{\ell'(1)}<0$. Hence $n_r'(1) = -\infty$ contradicting $n(q) < \ell(q)$ for $q$ close to 1.
\end{proof}

If $m_M<0$ but still $m(1)>0$ (and if $R<1$, $m'(0)<0$, $m'(1)>0$) and if $q_-$ is the root of $m$ in $(0,q_M)$ then the same proof gives that for $0 < r < q_-$, $\zeta(r)<1$.

\begin{lem} Suppose $q_M \notin \{ 0,1 \}$. Then
$\zeta(q_M)=q_M$.
\end{lem}

\begin{proof}
We have $n''(q) = \frac{\partial O}{\partial n} n'(q) +  \frac{\partial O}{\partial q}$. Then, at $(q,n(q)=m(q))$,
$n''(q) = \frac{\partial O}{\partial q}(q,m(q))=\frac{1-R}{R} \frac{m(q)}{(1-q)} \frac{m'(q)}{\ell(q)-m(q)}$.
If $R<1$, and $q_M \neq 1$ then $n_{q_M}(q_M) = m(q_M) = m_M$, $n'_{q_M}(q_M)= m'(q_M) = 0$ and $n_{q_M}''(q_M) = 0 < m''(q_M)$. Hence $n_{q_M}(q)$ lies below $m(q)$ to the right of $q_M$ and $\zeta(q_M)=q_M$.
\end{proof}

\begin{lem}
$\Lambda(q_M)=0$, $\lim_{r \downarrow 0} \Lambda(r) = \infty$ and $\Lambda$ is continuous and strictly decreasing.
\end{lem}

\begin{proof}
The strict monotonicity of $\Lambda$ follows from \eqref{eq:Lamderiv} and the fact the solutions $n_r$ are monotonic in $r$. Also, $\Lambda(q_M)=0$ is immediate from the fact that $\zeta(q_M)=q_M$. It remains to show that $\Lambda(0+) = \infty$. We prove this in the case $R<1$ and $\epsilon>0$, but the result follows similarly in other cases.

Let $\chi$ be the negative root of $H(x)=0$ where
\[ H(x) = Rx^2 - (1-R) \left[ \frac{\delta^2 R}{2} - R \epsilon + 1 \right] x - \epsilon (1-R)^2 \]
It is easy to see that $H(-\epsilon(1-R))> 0 > H( (1-R)(\frac{\delta^2}{2}-\epsilon))$ and hence
$m'(0) = - \epsilon(1-R) < \chi < (\ell'(0))^+ = (1-R)(\frac{\delta^2}{2}-\epsilon) \wedge 0$.
In fact if $n_0$ is differentiable at zero, then by l'H\^{o}pital's rule, $n_0'(0)$ solves
\[ n_0'(0) = \frac{(1-R)}{R} \frac{m'(0) - n'_0(0)}{\ell'(0) - n_0'(0)} \]
so that $\chi$ has the interpretation of a candidate value for $n_0'(0)$.

Fix $\rho$ with $m'(0) < \rho < \chi<0$. Then $H(\rho)>0$. Let $b(q)=1+ \rho q$. Then
\begin{eqnarray*}
O(q,b(q)) - \rho & = & -\frac{(1-R)}{R} \frac{(1 + \rho q)}{(1-q)} \frac{[(1+\rho q) - m(q)]}{[\ell(q) - (1 + \rho q)]} - \rho \\
 & = & -\frac{(1-R)}{R} \frac{(1 + \rho q)}{(1-q)}
 \frac{[\rho  + \epsilon(1-R) - \frac{\delta^2}{2}R(1-R)q]}{[(\frac{\delta^2}{2} - \epsilon)(1-R) - \rho  - \frac{\delta^2}{2}(1-R)^2q]} - \rho \\
 & = & \frac{1}{R(1-q) [(\frac{\delta^2}{2} - \epsilon)(1-R) - \rho  - \frac{\delta^2}{2}(1-R)^2q] } \left[ H(\rho) - B(\rho) q \right] \\
\end{eqnarray*}
where $B=B(\rho)$ is the constant
\[ B = (1-R)\rho (\rho  + \epsilon(1-R)) - \frac{\delta^2}{2}R(1-R)^2 - R \rho\left\{ \left( \frac{\delta^2}{2} - \epsilon \right)(1-R) - \rho \right\} - \rho \frac{\delta^2}{2}R(1-R)^2. \]
Since $H(\rho)>0$ by construction, there exists $Q^{(1)}>0$ such that for $0< q < Q^{(1)}$ we have $O(q,b(q))> \rho > m'(0)$.

For $r<Q^{(1)}$ let $\psi(r) = \inf \{ q : n_r(q) \geq b(q) \}$. If $n_r$ crosses $b$ before $Q^{(1)}$ then it crosses from below and stays above $b$ until $Q^{(1)}$. Also, for $r < q < Q^{(1)}\wedge \psi(r)$ we have $n_r'(q) = O(q, n_r(q)) > O(q, b(q)) > \rho$ by the monotonicity in the second argument of $O$.

Since $\rho > m'(0)$ there exists $Q^{(2)}>0$ such that for $0< q < Q^{(2)}(q)$ we have $m'(0) < m'(q) < \frac{m'(0) + \rho}{2}$. Then for $r< q< Q^{(1)}\wedge Q^{(2)} \wedge \psi(r)$,
\[ n'_r(q) - m'(q) > \rho - \frac{m'(0) + \rho}{2} = \frac{\rho - m'(0)}{2} =: \hat{\rho} \]
and for $r \leq q \leq  Q^{(1)}\wedge Q^{(2)} \wedge \psi(r)$,
\[ n_r(q) - m(q) \geq \hat{\rho}(q-r). \]
Further, for $\psi(r)<q< Q^{(1)}\wedge Q^{(2)}$  we have
$m(q) < 1 + ( \frac{m'(0)+\rho}{2} ) q$ and
\[ n_r(q) - m(q) > 1 + \rho q - 1 - \frac{m'(0)+\rho}{2} q = \hat{\rho} q > \hat{\rho}(q-r) .\]

Recall that for $q \in (r, \zeta(r) \wedge 1)$,  $\ell(q) - n_r(q) < \ell(q)-m(q) = \frac{\delta^2}{2}(1-R)q(1-q)$. Then, for $0<r<q < \tilde{q}$  where $\tilde{q}:=Q^{(1)}\wedge Q^{(2)} \wedge q_M$ we have
\[ \Gamma(r) > \int_r^{\tilde{q}} \frac{dq}{q(1-q)} \frac{\hat{\rho}(q-r)}{q(1-q) \frac{\delta^2}{2}(1-R)} > \frac{2 \hat{\rho}}{\delta^2(1-R)} \int_r^{\tilde{q}} dq \frac{(q-r)}{q^2} . \]
But $ \int_r^{\tilde{q}} dq \frac{(q-r)}{q^2} = \ln (\tilde{q}/r) + \frac{r}{\tilde{q}}-1$ which diverges as $r \downarrow 0$.

\end{proof}

\subsection{The threshold value of transaction costs below which the problem is ill-posed}
\begin{prop}
\label{prop:integral}
Let $Q$ and $P$ be the quadratics $Q(q)=(q_+ - q)(q - q_-)$ and $P(q) = (p_+ - q)(q - p_-)$. Suppose
either $p_- < 0 < q_- < q_+ < 1 < p_+$ or $p_- < 0 < 1 < p_+ < q_- < q_+$ or $q_- < q_+ < p_- < 0 < p_+$.

Then
\begin{eqnarray*}
\int_{q_-}^{q_+} dq \frac{1}{q(1-q)} \frac{Q(q)}{P(q)}
& = & \frac{q_+ q_-}{p_+ p_-} \ln \frac{q_+}{q_-}
-  \frac{(1-q_+)(1 -q_-)}{(p_+-1)(1-p_-)} \ln \frac{1-q_-}{1-q_+} \\
&& \hspace{5mm}
+ \frac{(p_+ - q_+)(p_+ - q_-)}{p_+(p_+ - 1)(p_+ - p_-)} \ln \frac{p_+- q_-}{p_+ - q_+} \\
&& \hspace{5mm}
- \frac{(q_+ - p_-)(q_- - p_-)}{p_-(1-p_-)(p_+ - p_-)} \ln \frac{q_+ - p_-}{q_- - p_-} .
\end{eqnarray*}
\end{prop}

\begin{proof}
We have 
\begin{eqnarray*}
\frac{(q_+-q)(q-q_-)}{q(1-q)(p_+-q)(q-p_-)} & = &
\frac{q_+ q_-}{p_+ p_-} \frac{1}{q}
- \frac{(1-q_+)(1 -q_-)}{(p_+-1)(1-p_-)} \frac{1}{1-q} \\
&& \hspace{5mm}
+ \frac{(p_+ - q_+)(p_+ - q_-)}{p_+(p_+ - 1)(p_+ - p_-)} \frac{1}{p_+ - q}
- \frac{(q_+ - p_-)(q_- - p_-)}{p_-(1-p_-)(p_+ - p_-)} \frac{1}{q - p_-} .
\end{eqnarray*}
The result follows on integrating. Note that under the relationships between $q_\pm$ and $p_\pm$ given in the statement of the proposition there are no roots of $P$ in $[q_-,q_+]$ and each of the four logarithms has a positive argument.
\end{proof}

\begin{cor}
\label{cor:underLam}
Suppose $|\epsilon|>\delta \sqrt{\frac{2R}{1-R}}$.
Then
\begin{eqnarray*}
\underline{\Lambda} = \int_{q_-}^{q_+} dq \frac{1}{q(1-q)} \frac{|m(q)|}{\ell(q)}
& = & - \ln \frac{q_+}{q_-}
-  \ln \frac{1-q_-}{1-q_+} \\
&& \hspace{5mm}
+ \frac{R}{1-R}\frac{(p_+ - q_+)(p_+ - q_-)}{p_+(p_+ - 1)(p_+ - p_-)} \ln \frac{p_+- q_-}{p_+ - q_+} \\
&& \hspace{5mm}
- \frac{R}{1-R} \frac{(q_+ - p_-)(q_- - p_-)}{p_-(1-p_-)(p_+ - p_-)} \ln \frac{q_+ - p_-}{q_- - p_-}
\end{eqnarray*}
where $q_\pm$ are the roots of $m$ and $p_\pm=$ are the roots of $\ell$:
\[ q_\pm =\frac{\epsilon \pm \sqrt{\epsilon^2 - (\delta \sqrt{\frac{2R}{1-R}})^2}}{\delta^2 R}
\hspace{20mm}
p_\pm = \frac{\frac{1}{2}\delta^2 - \epsilon \pm \sqrt{2 \delta^2 + (\frac{1}{2} \delta^2 - \epsilon)^2}}{\delta^2(1-R)} \]
\end{cor}

\begin{proof}
We have $m(q) = -\frac{\delta^2}{2}R(1-R)(q-q_-)(q_+-q)$ and $\ell(q) = \frac{\delta^2}{2}(1-R)^2(q-p_-)(p_+-q)$ and note $m(0)=1=\ell(0)$ and $\ell(1)=m(1)$.
Then with $Q(q)$ and $P(q)$ as in Proposition~\ref{prop:integral},
\begin{equation}
\int_{q_-}^{q_+} dq \frac{1}{q(1-q)} \frac{|m(q)|}{\ell(q)}
 =  \frac{R}{1-R} \int_{q_-}^{q_+} dq \frac{1}{q(1-q)} \frac{Q(q)}{P(q)}
\end{equation}
The result follows using $q_+ q_- = \frac{2}{\delta^2 R(1-R)}$ and $p_+ p_- = \frac{2}{-\delta^2(1-R)^2}$ so that $\frac{q_+ q_-}{p_+ p_-} = -\frac{1-R}{R}$ and
$(1-q_+)(1- q_-) = \frac{2m(1)}{\delta^2 R(1-R)}$ and $(p_+-1)(1-p_-) = \frac{2 \ell(1)}{\delta^2(1-R)^2}$ so that $\frac{(1-q_+)(1- q_-)}{(p_+-1)(1-p_-)} = \frac{1-R}{R}$.
\end{proof}

\section{Singular point of $n$}
Our goal is understand the nature of solutions which pass through the singular point $(1,m(1))$ and to prove Lemma~\ref{lem:singular} and Lemma~\ref{lem:Lambda}.

We assume that $\delta^2 R < \epsilon$ and also if $R<1$ that $\epsilon < \frac{1}{1-R} + \frac{\delta^2 R}{2}$. Then $(1-R)m'(1)<0$ and $m(1)>0$. We are interested in the behaviour of $n$ as it passes through the singular point $(1,m(1))$.

Let $\eta(x) = \frac{n(1+x)-m(1+x)}{\frac{1}{2} \delta^2(1-R)}$. Then the singular point is now at the origin. We have
\begin{eqnarray}
 \eta'(x) & = &  \frac{2}{\delta^2 R} n(1+x) \left[ - \frac{1}{x} - \frac{ \frac{1}{2} \delta^2 (1-R) (1+x) }{( m(1+x) - \frac{1}{2}\delta^2(1-R) x(1+x) - n(1+x) )} \right] - \frac{2m'(1+x)}{\delta^2(1-R)} \nonumber \\
 & = & - \frac{a(x,\eta)}{x^2} \eta + b(x) \label{eq:etaode}
\end{eqnarray}
where
\begin{eqnarray*}
a(x,\eta) & = & 
\frac{2}{\delta^2 R} \frac{m(1) + (1-R)(\delta^2 R - \epsilon)x + \frac{1}{2} \delta^2 R(1-R)x^2 + \frac{1}{2} \delta^2 (1-R) \eta}{1 + x + \frac{\eta}{x} }\\
b(x) & = &  - \frac{2}{\delta^2 (1-R)} m'(1+x) = \frac{2}{\delta^2} \left[{\epsilon} - {\delta^2 R} - \delta^2 R x \right]
\end{eqnarray*}

We have $m(1)>0$ and $m'(1)<0$ whence $\lim_{x \downarrow 0} a(x,0)=\frac{2}{\delta^2 R} m(1)>0$ and $b(0)>0$. Note that for sufficiently small $x$ we have that $a$ is positive and decreasing in the second argument and $b$ is positive. Our focus is on proving results about existence and uniqueness of solutions in a neighbourhood of $x=0$.

For intuition, and following Choi et al~\cite[Lemma 6.8]{ChoiSirbuZitkovic:13} consider the initial value ODE
\begin{equation}
\label{eq:intuition}
 f' = - A \frac{f}{x^2} + B,  \hspace{20mm} f(0)=0 ;
\end{equation}
where $A$ and $B$ are positive constants. We look for a solution $f = f(x)$ separately in $x \geq 0$ and $x \leq 0$. We find that there are multiple solutions for $x \leq 0$, but a unique solution for $x \geq 0$.

Fix $y>0$ and define $D(x) = D_y(x)$ by $D(x) = \int_x^y \frac{A}{z^2} dz = A \left( \frac{1}{x} - \frac{1}{y} \right)$. Then $D(0+)=\infty$. Also $\frac{d}{dx} (e^{-D(x)} f(x)) = Be^{-D(x)}$ and hence $f(x) = e^{D(x)} (f(y) - \int_x^y B e^{-D(z)} dz)$. The condition $f(0)=0$ forces $f(y) = \int_0^y B e^{-D(z)} dz$ and then
\[ f(x) = \int_0^x B \exp \left( - \int_z^x \frac{A}{y^2} dy \right) dz.  \]
Thus, the solution $f$ is unique for $x \geq 0$.

Now we look for a solution in $x \leq 0$. Fix $y<0$ and set $D(x)=D_y(x) = \int_y^x \frac{A}{z^2} dz$. Then $f(x) = e^{-D(x)} (f(y) - \int_y^x B e^{D(z)} dz) = e^{-D(x)} f(y) - \int_y^x B e^{-(D(x) - D(z))} dz$. Each value of $f(y)$ leads to a solution $f$ for which $f(0)=0$

We can also analyse the behaviour near $x=0$ of solutions to \eqref{eq:intuition}.
We have $f(x) = B \int_0^x e^{-A [ \frac{1}{y} - \frac{1}{x} ]} dy$. Then
\[ f(x) \leq B  e^{\frac{A}{x}} \int_0^x e^{-\frac{A}{y}} \frac{x^2}{y^2} dy = \frac{B}{A} x^2  e^{\frac{A}{x}} \int_0^x e^{-\frac{A}{y}}\frac{A}{y^2} dy  = \frac{B}{A} x^2. \]
Conversely, integrating by parts
\begin{eqnarray*} f(x) & = & \frac{B}{A}e^{\frac{A}{x}} \int_0^x \frac{A}{y^2} e^{-\frac{A}{y}} y^2 dy = \frac{B}{A}e^{\frac{A}{x}} \left\{ \left[y^2 e^{-\frac{A}{y}} \right]^x_0 - 2 \int_0^x y e^{-\frac{A}{y}} dy \right\} \\
& = & \frac{B}{A} x^2 - \frac{2B}{A} e^{\frac{A}{x}} \int_0^x \frac{y^3}{A} \frac{A}{y^2} e^{-\frac{A}{y}} dy  \geq  \frac{B}{A} x^2 - \frac{2B}{A^2} x^3 e^{\frac{A}{x}} \int_0^x \frac{A}{y^2} e^{-\frac{A}{y}} dy = \frac{B}{A} x^2 - \frac{2B}{A^2} x^3
\end{eqnarray*}
In particular,
\begin{equation}
\label{eq:limit}
\lim_{x \rightarrow 0} x^{-2}f(x) = \frac{B}{A}.
\end{equation}
We remark that if $A=A(x)$ and $B=B(x)$ are continuous and positive at $x=0$, then a small extension of the above argument gives
that a solution to \eqref{eq:intuition} satisfies $\lim_{x \rightarrow 0} x^{-2} f(x) = \frac{B(0)}{A(0)}$.

Now we turn to the solution of the problem
\begin{equation}
\label{eq:real}
 \eta' = - a(x,\eta) \frac{\eta}{x^2} + b(x)  \hspace{20mm} \eta(0)=0 .
\end{equation}
We have seen that each member of the family $\{ \eta_r \}_{r \in (0,1)}$ with $\eta_r$  given by
\[ \eta_r(x) = \frac{n_r(1+x) - m(1+x)}{\frac{1}{2} \delta^2 (1-R) }  \hspace{20mm} -(1-r) \leq x \leq 0 \]
solves \eqref{eq:real} for $x \leq 0$. So, our focus is on the case $x \geq 0$. Our consideration of \eqref{eq:intuition} leads us to expect that there is a unique solution.

\begin{prop}
\label{prop:real}
There exists a unique solution to \eqref{eq:real} in $x \geq 0$.
\end{prop}

\begin{proof}
Away from $x=0$ standard theory (see for example Walter~\cite[Chapter II, Section 7]{Walter:98}) gives the existence and uniqueness of a solution passing through any point $(x,\xi)$. So, our focus is on solutions near the origin.

For $x$ small enough $a(x,\xi)$ is positive and decreasing in $\xi$. We work on an interval $J^1=[0,x_1]$ such that $a(x,\xi)$ is positive and decreasing in $\xi$ on $J^1_0 = (0,x_1]$ and $b$ is positive and bounded on $J^1$. Further, we assume that for $x \in J^1_0$ we have the bounds $Cx \leq a(x,\xi) \leq a(x,0) =:A(x)$ and $ b(0)/2 \leq b(x) \leq 2b(0)$. Here $C = \frac{1-R}{R}$ where temporarily we assume $R<1$.

Define $L(x,g,h) = b(x) - \frac{a(x,g)}{x^2}h$. We work with functions defined on $J^1 = [0,x_1]$. Set $f_0(x) = 0$ and let $f_1$ be the solution to $f' = L(x,f_0,f)$ subject to $f(0)=0$. Then
\[ f_1(x) = \int_0^x b(y) \exp \left( - \int_y^x \frac{A(z)}{z^2} dz \right) dy. \]
Now construct a sequence of differentiable functions $(f_{n})_{n \geq 0}$ on $J^1$ where $f_{n+1}$ solves
$f = L(x, f_n,f)$ subject to $f(0)=0$. Then
\[ f_{n+1}(x) = \int_0^x b(y) \exp \left( - \int_y^x \frac{a(z,f_n(z))}{z^2} dz \right) dy. \]
We argue that this family of solutions is increasing in $n$.
Clearly $f_1>0 = f_0$ on $J^1_0$. Suppose inductively that $f_n > f_{n-1}$ on $J^1_0$. Then, since $a$ is decreasing in its second argument
$a(z,f_{n}(z)) < a(z, f_{n-1}(z))$ and for $x \in J^1_0$
\[  \int_0^x b(y) \exp \left( - \int_y^x \frac{a(z,f_n(z))}{z^2} dz \right) dy
> \int_0^x b(y) \exp \left( - \int_y^x \frac{a(z,f_{n-1}(z))}{z^2} dz \right) dy.  \]
It follows that $f_{n+1}(x) > f_{n}(x)$ as required.

Now we look for an upper bound.
Since $a(x,f) \geq Cx > 0$,
\begin{equation}
\label{eq:bound}
f_n(x) \leq \int_0^x b(y) e^{-\int_y^x \frac{C}{z} dz} = \int_0^x b(y) \frac{y^C}{x^C} dy \leq
2 b(0) \frac{x}{C+1}
\end{equation}
and for $z<x$
\[  - 2 b(0) z \left[ 1 - e^{-\int_z^x \frac{C}{u} } du \right] \leq f_n(x)-f_n(z) \leq 2 b(0) (x-z). \]
It follows that the function $f$ on $J$ given by $f(x)= \lim_n f_n(x)$ exists and is continuous with $f(0)=0$.
It remains to show that $f$ solves \eqref{eq:real}.
We have
\[ 
f(x) = \lim_n f_n(x)  =  \lim_n \int_0^x b(y) \exp \left( - \int_y^x \frac{a(z,f_n(z))}{z^2} dz \right) dy
 =  \int_0^x b(y) \exp \left( - \int_y^x \frac{a(z,f(z))}{z^2} dz \right) dy 
\] 
by monotone convergence. Since the right-hand-side of this expression is continuous differentiable we have that $f$ is continuously differentiable and $f'=L(x,f,f)$ as required.

The only place that we use $R<1$ is to say that there is a positive lower bound for $a(x,\eta)$.
If $R>1$, then we can construct a solution not in the strip $J^1 \times \R$, but rather until it first leaves the rectangle
$[0,x_1] \times [0, K]$ for some bound $K$, where $K$ is chosen so that $a(x,K) > \frac{x}{K}$ on $J^1$. The inequality $a(x,f) > \frac{x}{K}$ on $J^1 \times [0,K]$ gives an upper bound like \eqref{eq:bound} on $J^1$  which again can be used to prove existence of the limit $f$. Setting $x_2 = \min \{ x_1, \frac{K+1}{2 b(0)} \}$ and $J^2=[0,x_2]$ we have $f_n(x) \leq K$ on $J^2$. Then restricting attention to $J^2$ we find $f' = L(x,f,f)$ as required.

Now we consider uniqueness.
Let $f$ and $g$ be solutions which are non-negative on $J^3=[0,x_3]$ for some $x_3 \leq x_2$. We have that $a(x,f(x))>0$ and $a(x,g(x))>0$ on $J^3_0 = (0,x_3]$. If $f(z)=g(z)$ at some $z$ in $J^3_0$ then since $a(x,f)\frac{f}{x^2}$ is Lipschitz in $f$ away from $x=0$ we have that $f=g$ on $J^3_0$ and hence on $J^3$.

So suppose $f(z)>g(z)$ on $J^3_0$. We want to show that this leads to a contradiction. Since $\xi a(z,\xi)$ is increasing in $\xi$ (for small enough $\xi$), if $f>g$ we have $0 <a(z,g)g < a(z,f)f$. Let $h=f-g$; then $h$ solves
\[ h' = -a(x,f) \frac{f}{x^2} + a(x,g) \frac{g}{x^2} < 0\]
and $0 < h(x) < h(0)=0$ which is the desired contradiction.
\end{proof}

We can give an asymptotic analysis of the solution $\eta$ to \eqref{eq:real} in the same spirit as \eqref{eq:limit}. Take $f_0 = 0$, then $\eta(x) \geq f_1(x) \geq \int_0^x b(y) e^{ - \int_y^x \frac{A(z)}{z^2} } dz$. It follows by the comment after \eqref{eq:limit} that
$\lim_{x \downarrow 0} x^{-2} f_1(x) = \frac{b(0)}{A(0)}$ where $A(0)= \lim_{x \downarrow 0} a(x,0)$.
Hence $\lim_{x \downarrow 0} x^{-2} \eta(x) \geq \frac{b(0)}{A(0)}$. Conversely, using $h_0(x) = 2\frac{b(0)}{C+1}x$ we can conclude
$\eta(x) \leq h_1(x)$ where $h_1$ solves $h' = L(x, h_0,h)$. It can be shown that $h_1(x) \leq \kappa x^2$ for some constant $\kappa$. Repeating the argument, if $h_2$ solves $h' = L(x,h_1,h)$ subject to $h(0)=0$, then $\eta \leq h_2$ and $\lim_{x \downarrow 0} x^{-2} h_2(x) = \frac{b(0)}{A(0)}$. Hence $\lim_{x \downarrow 0} x^{-2} \eta(x) = \frac{b(0)}{A(0)}$. As a byproduct we conclude $\eta'$ is well defined at zero and $\eta'(0)=0$.

\begin{proof}[Proof of Lemma~\ref{lem:singular}]

(i) This part of the lemma follows from Proposition~\ref{prop:real}, and the argument above that $\eta'(0)=0$..

(ii) Suppose $R<1$. The case $R>1$ is similar, but sometimes involves reversed inequalities. Since $m(q) \leq n_r(q) \leq \ell(q)$ for $r < q <1$ and $m(1)=\ell(1)$ it follows
that for $r<1$ we have $n_r(1)=m(1)$ and $\zeta(r)>1$.

We have $m'(1)<0$. Suppose it is not the case that $n_r'(1-) = m'(1)$. Then either there exists $q_0 \in (0,1)$, $\theta \in (\ell'(1), m'(1))$ such that $n(q) > m(1) - (1-q)\theta$ or there exists $q_k \uparrow 1$, $\theta \in (\ell'(1), m'(1))$ such that $n(q_k) = m(1) - (1-q_k)\theta$. In the former case
\[ (1-q)O(q, n(1-q)) < -\frac{(1-R)}{R} m(1) \frac{\theta}{\delta^2(1-R)} \]
and hence $n'(1-)=-\infty$ contradicting $n < \ell$ on $(0,1)$. In the latter case
\[ O(q_k, n(q_k)) = -\frac{(1-R)}{R} \frac{m(1) - (1-q_k)\theta}{q_k}  \frac{\theta}{\delta^2(1-R)(1-q_k)} < \theta \]
for large enough $k$. Hence, sufficiently close to 1, $n$ can only cross the line $m(1) + (1-q)\theta$ from above to below, contradicting the existence of a sequence $q_k \uparrow 1$.

(iii) This follows from the uniqueness of solutions to the right of the singular point.
\end{proof}

\begin{proof}[Proof of Lemma~\ref{lem:Lambda}]
Given the results in Lemma~\ref{lem:3.1} which transfer to this context, all we need to show is that $\Lambda$ is continuous at $r=1$. This will follow if $I_1$ and $I_2$ are finite for small positive $x$ and $r<1-x$ where
\[ I_1(x) =  \int_1^{(1+x)\wedge \zeta(1)} dq \frac{1}{q(q-1)} \frac{n_1(q) - m(q)}{n_1(q) - \ell(q)}
\hspace{10mm}
I_2(x) = \int^1_{1-x} dq \frac{1}{q(1-q)} \frac{n_r(q) - m(q)}{\ell(q) - n_r(q)}   \]
But $n_1(q)-m(q) \sim (q-1)^2$, whereas $n_1(q)-\ell(q) \sim (q-1)$, where we write $f(x) \sim x^\alpha$ if $\lim_{x \downarrow 0} x^{-\alpha} f(x) = C$ for $C \in \R \setminus \{ 0 \}$. Hence $I_1$ is finite. Similar expansions hold to the left of 1 and we conclude $I_2$ is finite.
\end{proof}

\end{document}